%% file: paper.tex
\documentclass[9pt,journal,compsoc]{IEEEtran}
\IEEEoverridecommandlockouts

\ifCLASSOPTIONcompsoc
  \usepackage[nocompress]{cite}
\else
  \usepackage{cite}
\fi

\title{Codes for Metastability-Containing Addition
}

\author{\IEEEauthorblockN{Johannes Bund\IEEEauthorrefmark{1}, \and Christoph Lenzen\IEEEauthorrefmark{2}, \and Moti Medina\IEEEauthorrefmark{3}}\\\vspace{1em}
\IEEEauthorblockA{\IEEEauthorrefmark{1}\textit{Université Paris-Saclay, CNRS, ENS Paris-Saclay, Laboratoire Méthodes Formelles}\\
91940 Gif-sur-Yvette, France \\
jobund@lmf.cnrs.fr}\\\vspace{.5em}
\IEEEauthorblockA{\IEEEauthorrefmark{2}\textit{CISPA Helmholtz Center for Information Security}\\
 66123 Saarbr\"ucken, Germany \\
 lenzen@cispa.de} \\\vspace{.5em}
 \IEEEauthorblockA{\IEEEauthorrefmark{3}\textit{Alexander Kofkin Faculty of Engineering, Bar-Ilan University}\\
5290002 Ramat Gan, Israel \\
moti.medina@biu.ac.il}
\thanks{
{This work has been accepted for publication at IEEE Transactions on Computers.
\emph{Funding:} This research was supported by the Israel Science Foundation under Grants
867/19 and 554/23.}
{This research was funded in part
by the French Agence Nationale de la Recherche (ANR) under the project DREAMY (ANR-21-CE48-0003).}
{Johannes Bund was affiliated with CISPA Helmholtz Center for Information Security and Bar-Ilan University when working on the technical parts of this paper.}
{During parts of this work Christoph Lenzen was affiliated with Aalto University.}
}}

\usepackage{amsmath}
\usepackage{amsthm}
\usepackage{amssymb}
\usepackage{mathtools}

\usepackage{relsize}
\usepackage{multirow}

\usepackage{hyperref}
\usepackage[capitalize]{cleveref}
\usepackage{comment}
\usepackage{soul}
\usepackage{color}

\usepackage{tikz}
\usetikzlibrary{automata,positioning}

\usepackage{colortbl}
\definecolor{Gray}{gray}{.8}

\graphicspath{{./figures/}}

\newtheorem{theorem}{Theorem}[section]
\newtheorem{corollary}[theorem]{Corollary}
\newtheorem{definition}[theorem]{Definition}

\newtheorem{lemma}[theorem]{Lemma}
\newtheorem{observation}[theorem]{Observation}

\theoremstyle{remark}
\newtheorem*{remark}{Remark}

\DeclareMathOperator{\res}{res}
\DeclareMathOperator{\map}{map}

\DeclareMathOperator{\parity}{par}

\newcommand{\ovf}{\textsf{ovf}}
\newcommand{\starr}{*}
\newcommand{\bigstarr}{\mathop{\raisebox{-.7pt}{\ensuremath{\mathlarger{\mathlarger{\mathlarger{*}}}}}}}
\newcommand{\mfu}{\mathtt{M}}
\newcommand{\mfx}{\mathtt{X}}

\newcommand{\IN}{\mathbb{N}}
\newcommand{\IB}{\mathbb{B}}
\newcommand{\IT}{\mathbb{T}}

\newcommand{\AND}{\textsf{and}}
\newcommand{\XOR}{\textsf{xor}}
\newcommand{\OR}{\textsf{or}}
\newcommand{\NOT}{\textsf{not}}
\newcommand{\MUX}{\textsf{MUX}}

\newcommand{\ceil}[1]{\lceil{#1}\rceil}
\newcommand{\floor}[1]{\lfloor{#1}\rfloor}
\newcommand{\ang}[1]{\langle{#1}\rangle}

\newcommand{\intvl}[1]{\ang{#1}}

\newcommand{\gBin}[1]{\gamma_{#1}^{\operatorname{b}}}

\newcommand{\gBRGC}[1]{\gamma_{#1}^{\operatorname{g}}}

\newcommand{\gUnary}[1]{\gamma_{#1}^{\operatorname{u}}}
\newcommand{\gbUnary}[1]{\overline{\gamma}_{#1}^{\operatorname{u}}}
\newcommand{\gtUnary}{(\tilde{\gamma}^{\operatorname{u}})^{-1}}
\newcommand{\gtHyb}{(\widetilde{\gamma}^{\operatorname{h}})^{-1}}

\newcommand{\gHyb}[1]{\gamma_{#1}^{\operatorname{h}}}
\newcommand{\mUnary}[1]{\map_{#1}^{\operatorname{u}}}

\newcommand{\BLM}{{\textsf{ADD}}_{n,k}}

\newcommand{\lmax}[1]{\ell^{\max}_{#1}}
\newcommand{\lmin}[1]{\ell^{\min}_{#1}}

\newcommand{\BO}{\mathcal{O}}

\allowdisplaybreaks

\begin{document}

\maketitle
\begin{abstract}

\input{abstract}
\end{abstract}
\input{intro}
\input{definitions}
\input{lower_bound}
\input{codes}

\input{old_code}
\input{addition}
\input{related}
\input{conclusion}

\ifCLASSOPTIONcompsoc
 \section*{Acknowledgments}
\else
 \section*{Acknowledgment}
\fi
We thank Matthias F\"ugger for valuable discussions.

\ifCLASSOPTIONcaptionsoff
  \newpage
\fi

\bibliographystyle{IEEEtran}
\bibliography{IEEEabrv, comp}

\end{document}

%% file: abstract.tex
We investigate the fundamental task of addition under uncertainty, namely, addends that are represented as intervals of numbers rather than single values. One potential source of such uncertainty can occur when obtaining discrete-valued measurements of analog values, which are prone to metastability. Naturally, unstable bits impact gate-level and, consequently, circuit-level computations.  
Using Binary encoding for such an addition produces a sum with an amplified imprecision. 


Hence, the challenge is to devise an encoding that does not amplify the imprecision caused by unstable bits. We call such codes \emph{recoverable.} 
While this challenge is easily met for unary encoding, no suitable codes of high rates are known.    
%
%
In this work, we prove an upper bound on the rate of preserving and recoverable codes for a given bound on the addends' combined uncertainty. We then design an asymptotically optimal code that preserves the addends' combined uncertainty.

We then discuss how to obtain adders for our code. The approach can be used with any known or future construction for containing metastability of the inputs. We conjecture that careful design based on existing techniques can lead to significant latency reduction.


%% file: intro.tex
\section{Introduction}\label{sec:intro}

At first glance, addition of inputs carrying uncertainty is straightforward:
representing the inputs as intervals, one simply adds the lower and upper bounds of these input intervals, respectively, yielding an interval specifying the possible outputs.
However, this drastically changes when instability of inputs interferes with computation.
For example, such interference can be caused by metastability.
Metastable signals disrupt the binary representation of inputs, making it impossible to represent the lower and upper bounds of an input interval.
Previous results show that metastability-containing circuits can mitigate the loss of information due to metastable signals~\cite{friedrichs18containing}.
We say that a combinational circuit is \emph{metastability-containing (mc)} if it does not deviate from the best possible response to inputs with metastable bits (for a formal definition, see Section~\ref{sec:meta}). 
In \cite{bund2025small}, we observe that the choice of encoding matters when facing metastable inputs.
In this work, we show that the careful design of encoding and metastability-containing addition circuits makes it possible to maintain the implicit lower and upper bounds of input intervals, preventing the loss of information.

In the following example, we show that \emph{encoding matters} when facing metastable bits, i.e., even if one uses a metastability-containing binary adder, then information is still lost. 

\begin{table}
  \caption{Addition of binary codewords, using a standard VHDL implementation of a serial adder, with input signals $(\operatorname{X}, \operatorname{Y})$ and output signal ($\operatorname{SUM}$). Metastable signals are represented by the value $\mfx$.}
  \label{tab:simu_bin}
  \centering
  \begin{tabular}{rll}
    $\operatorname{X}$ & $0001\,1001$ $(\{25\})$ & $0001\,101\mfx$ $(\{26,27\})$ \\
    $\operatorname{Y}$ & $0010\,0101$ $(\{37\})$ & $0010\,0101$ $(\{37\})$ \\ \hline
    $\operatorname{SUM}$ & $0011\,1110$ $(\{62\})$ & $0\mfx\mfx\mfx\,\mfx\mfx\mfx\mfx$ $(\{0,\hdots,127\})$ 
  \end{tabular}
\end{table}

\subsubsection*{Addition Under Uncertainty: First Example}
A classic example of loss of information due to metastability
  is addition under standard binary encoding.
In \cref{tab:simu_bin} we simulate a classical combinational adder implemented in VHDL. For stable inputs (binary encodings of $25$ and $37$), the circuit produces the correct output (binary encoding of $62$). However, for a single metastable input bit, the circuit amplifies the uncertainty. When switching between the encoding of $26$ and $27$, the least significant bit may become metastable. 
This example causes $7$ out of $8$ output bits to become metastable.
Resolving each $\mfx$ to either $0$ or $1$ can produce any codeword from $0$ to 
  $127$.

Let us consider a naive solution where Gray code is used, which has at most one metastable bit in each codeword, representing intervals of size $2$. These arise when using Time-to-Digital converters (TDCs) to generate the addends, as these can be designed to deterministically ensure at most one metastable bit in a Gray code output~\cite{fuegger17aware}.
In this case, adding two or more values with a single metastable bit each must aggregate the corresponding uncertainty. In turn, uncertainty is again increased, e.g., adding $0\mfx1$ and $11\mfx$ that corresponds to $\{1,2\}$ and $\{4,5\}$ yields $1\mfx\mfx$ which corresponds to $\{4,\ldots, 7\}$, instead of the desired $\{5,6,7\}$. 

Informally, we use the term \emph{precision} to describe how sharply an encoded value is determined: the fewer consecutive integer values it may represent, the higher its precision. We parameterize limited precision by an \emph{imprecision} $p \geq 0$: a codeword with imprecision $p$ may
resolve to at most $p+1$ consecutive integer values. In the following sections, we make this notion gradually more formal.

\subsubsection*{Our Main Research Question}
Arguably, the main challenge we overcome lies in formalizing the concept of addition under uncertainty by
introducing the notion of ``precision'' and pitching it against a code's rate.
Intuitively, the precision of an input is the range of numbers it represents; for instance, the range $\{7,8,9
\}$ has precision~$2$.
A suitable encoding must preserve this precision in the sense that after metastability is resolved, we can still
reliably recover a number from this range.

In prior work (cf.~\cite{bund20optimal} and \cite{bund2025small}), we observe that the amount of imprecision that can be tolerated is restricted by the
  encoding of the inputs.
In the Boolean world, the choice of encoding may influence the complexity
  of an operation, but not its precision.
In the face of metastability, however, the encoding might cause unnecessary loss of
  information (cf.\ Table~\ref{tab:simu_bin}).
In classic designs, the consequence is that synchronizers must ensure that \emph{no} metastability is propagated into the addition circuit.

Note that an ``inefficient'' precision-preserving encoding and a corresponding addition circuit are attainable.  
To see this, consider the simple solution based on a unary thermometer code combined with an mc sorting circuit.
The unary thermometer code preserves precision, and mc sorting of the union of the input words' bits implements the addition. 
Optimal mc sorting circuits are known~\cite{bund20optimal}.
Therefore, this naive mc solution is efficient in $n$. 
However, this is misleading: using $n$ bits, it can only add numbers from the range $\{0,\ldots,n\}$, while a classic adder operates on exponentially smaller encodings, i.e., handles the range $\{0,\ldots,2^n-1\}$ using the same number of bits.

Indeed, as we show later on (cf.~\cref{lem:recoverability}), without a bound on the 
  number of metastable bits, nothing can be done to improve this situation, because $n$ bits are required for a code preserving precision~$n$. 
Conversely, Gray codes have optimal rate and can handle inputs that correspond to adjacent numbers, e.g., $\{2,3\}$ or $\{62,63\}$ (cf.~\cref{obs:gray_code}). 
This gives rise to the following question. 
\begin{quote}
    \emph{Is there a high rate code that (i) preserves precision when adding codewords with a bounded number of metastable bits and (ii) allows for efficient implementation by a circuit?}
\end{quote}

\begin{table}
  \caption{Addition of hybrid codewords in VHDL, using a metastability-containing implementation of the addition scheme in 
           \cref{sec:addition}. Codewords are of length $8$, where $n=5$ and $k=3$. Metastable signals are represented by the value $\mfx$.}
  \label{tab:simu_mc}
  \centering
  \begin{tabular}{rll}
    $\operatorname{X}$ & $00101\,1\mfx0$ $(\{25,26\})$ & $01\mfx10\,\mfx00$ $(\{47,48,49\})$ \\
    $\operatorname{Y}$ & $01101\,011$ $(\{37\})$ & $00111\,011$ $(\{21\})$ \\ \hline
    $\operatorname{SUM}$ & $01000\,00\mfx$ $(\{62,63\})$ & $11001\,\mfx\mfx1$ $(\{68,69,70\})$ 
  \end{tabular}
\end{table}

\subsection{Our Contributions}\label{sec:contrib}
In this paper, we establish the following results. 

\medskip
\noindent
\textbf{Formulating Preserving and Recoverable Codes.}
We formalize the concept of addition under uncertainty by introducing the code properties \emph{preservation} and \emph{recoverability}. Informally, in the context of addition, we say that a code is preserving if the range of codewords that can be obtained by resolving metastability in the output of an mc addition circuit matches the sum of the addends' uncertainties. Furthermore, a preserving code is deemed recoverable if one can fix a mapping from \emph{all} strings to codewords such that after resolving metastability, applying this mapping yields a codeword from the above range.

\medskip
\noindent
\textbf{Upper Bound on The Rate of Recoverable Codes.}
We prove that the rate of an $n$-bit code preserving $\ell$-precision is bounded by $\BO(2^{n-\ell}\ell/n)$, where $\ell\in\IN$ is a bound on the imprecision. 

\medskip
\noindent
\textbf{An Asymptotically Optimal Code.}
We present an $(n+\ell)$-bit code that is $\ell$-precise.
In fact, the code is $\ell$-preserving and $\ceil{\ell/2}$-recoverable, 
indicating that intervals of size $\ell$ can be represented by a codeword and 
we can recover intervals of size $\ceil{\ell/2}$ from metastable signals
(see Definition~\ref{def:preserve} and Definition~\ref{def:recover}).

\medskip
\noindent
\textbf{MC Addition Circuit for Our Code.}
For the $\ell$-precise code we present an addition circuit of size $\BO(n+\ell)$ and depth $\BO(\log n + \log \ell)$.
However, this circuit implements addition only for stable codewords.
Two more steps are required: (i) mapping (stable) non-codewords to codewords, which is straightforward; and (ii) creating an mc circuit with the same behavior on stable inputs.
A purely combinational mc circuit of size $(n+\ell)^{\BO(\ell)}$ and depth $\BO(\ell\log n)$ is obtained using a construction from~\cite{ikenmeyer18complexity}.
Relying on so-called masking registers, a construction from~\cite{friedrichs18containing} yields a circuit of size $\BO(\ell(n + \ell))$ and depth $\BO(\log n + \log \ell)$.

\medskip
The following example shows how an mc adder of codewords from our code does not lose precision. 
\subsubsection*{Addition Under Uncertainty: Second Example}
In \cref{tab:simu_mc} we show the addition of codewords from our ``hybrid'' code (cf.~\cref{sec:hybcode}).
In the example of \cref{tab:simu_mc}, we use an $8$-bit hybrid encoding with a
5-bit binary reflected Gray code (BRGC) part and a $3$-bit unary variant. We start from the all-zero
word $00000\,000$ and define an up-counting sequence as follows. While the unary part is neither all $0$ nor all $1$, an up-count changes only the unary part, following the pattern
\[
000 \to 100 \to 110 \to 111 \to 011 \to 001 \to 000,
\]
so that consecutive words differ in a single unary bit. Whenever the unary part is $111$ or $000$, the next up-count flips only the BRGC
part (advancing one step in the 5-bit BRGC) and leaves the unary part unchanged; after this, the unary pattern continues as above. This yields a global Gray sequence on the $8$-bit codewords.

For the concrete integer values that appear in \cref{tab:simu_mc}, the corresponding hybrid codewords (with BRGC and unary parts separated by
a space) are:
\begin{align*}
21 &\mapsto 00111\ 011,
25 \mapsto 00101\ 100,
26 \mapsto 00101\ 110, \\
37 &\mapsto 01101\ 011,
47 \mapsto 01110\ 000,
48 \mapsto 01010\ 000, \\
49 &\mapsto 01010\ 100,
62 \mapsto 01000\ 001,
63 \mapsto 01000\ 000, \\
68 &\mapsto 11001\ 111,
69 \mapsto 11001\ 011,
70 \mapsto 11001\ 001.
\end{align*}
Here, the BRGC part determines the coarse position, and the unary part refines it within a small range; the integer intervals shown in \cref{tab:simu_mc} are obtained by interpreting these words via the hybrid
encoding.

Metastable signals are represented by the value $\mfx$.
The simulated circuit is a VHDL metastability-containing implementation of the circuit presented
  in \cref{sec:addition}. 
Adding an input that may represent $25$ or $26$ and the encoding 
  of $37$ correctly results in a codeword that may represent $62$ or $63$.
Our hybrid code is parameterized by $k$.
For $k=3$ the code allows imprecision up to $2$, i.e., a codeword that 
  represents $3$ consecutive numbers that can be added correctly.
For example, the addition of the codeword that represents $47$, $48$, 
  and $49$ plus the encoding of $21$ correctly results in the 
  codeword that represents $68$, $69$, and $70$.

\subsection{Dealing with Metastability: From Synchronizers to Metastability Containing Circuits}

Metastability is a spurious mode of operation in digital circuits.
Any clock domain crossing, analog-to-digital, or time-to-digital 
  conversion inherently bears the risk of metastable signals~\cite{ginosar11tutorial}.
The phenomenon involves bistable storage elements that are in a volatile state.
The output of a metastable storage element may be stuck at a voltage level between logical $0$ and
  logical $1$, it may be oscillating, or it can have late transitions. 
Metastable signals may break timing constraints and input specifications of gates
  and storage elements. 
Hence, they propagate through the digital circuit, possibly affecting other signals. 
Chances of circuit malfunction are usually denoted by the mean time between failures (MTBF).

By their nature, metastable signals cause soft errors with potentially critical 
  consequences.
At any point in time, a metastable signal may be read as logical $0$ or logical $1$.
Hence, metastability causes uncertainty regarding whether the encoded value is $0$ or $1$.
No digital circuit 
  can deterministically avoid, resolve, or detect metastability~\cite{marino1977effect}.

Trends in VLSI circuit technology further boost the risk of metastable upsets in storage elements.
Factors like increased complexity, higher clock frequency, more 
  unsynchronized clock domains, lower operating voltage, and smaller
  transistor technologies reduce the MTBF~\cite{ginosar2003fourteen,beer2013metastability}.

A common approach to lowering the risk of propagating metastability is to employ 
  synchronizers~\cite{kinniment08}.
A synchronizer is a device that amplifies the resolution of a metastable signal towards logical $0$
  or logical $1$.
Essentially, synchronizers trade time for reliability. 
The more time is allocated for the resolution of metastable signals, the smaller
  the probability of soft errors induced by metastability gets.
Even though achieving possibly large MTBF, errors 
  can never be avoided deterministically.
Addressing the risks by adding synchronizers increases synchronization
  delays and hence latency.
Therefore, it is desirable to handle the threat of metastable upsets by techniques without these drawbacks.

Interestingly, the study of switching networks shows that careful design
  can \emph{deterministically} avoid soft errors~\cite{huffman57design}, seemingly defying Marino's impossibility result~\cite{marino1977effect}. 
The key insight is that circuits can be designed to contain the uncertainty induced by
  metastable signals without amplification~\cite{friedrichs18containing};
  we call such circuits \emph{metastability-containing (mc)}.
An mc circuit deterministically minimizes the propagation of metastable signals.
That is, assuming worst-case propagation of metastability, it produces the most precise output possible.

Albeit avoidable in some applications~\cite{bund20optimal,fuegger17aware}, in general constructing mc circuits incurs an 
  exponential blow-up in the circuit size~\cite{ikenmeyer18complexity}.
  However, this overhead can be reduced to linear using so-called \emph{masking} storage elements, which translate internal metastability into late output transitions~\cite{friedrichs18containing}.
  Recently, practical high-quality masking latches have been devised~\cite{srinivas25latch}. 

\subsection{Application: Fault-tolerant Clock Synchronization}
A concrete application where metastability-containing addition naturally
arises is fault-tolerant hardware clock synchronization in the style of
Lynch and Welch~\cite{welch88fault_tolerant}, as implemented in a metastability-containing way in~\cite{friedrichs18containing}. There, time-to-digital
converters (TDCs) produce phase measurements in a precision-1 code
(typically Gray or unary), i.e., each measurement may differ from the
true integer by at most one and contains at most one metastable bit.
In a basic step of the Lynch-Welch algorithm, each node collects a
number of such phase measurements, discards outliers,
sums the remaining values, and then divides the sum by $2$ to obtain a small correction that is applied to its local clock.

The computational core therefore needs to add two encoded
phase measurements, each with at most one metastable bit. The total
input imprecision of the addition is at most $k=2$, and we can work with the corresponding hybrid code (referred to as $\gHyb{n,3}$ in Section~\ref{sec:hybcode}), which uses only $2$ redundant bits, i.e., requiring little overhead compared to
a standard $n$-bit Gray code for the concatenated unary encoding variant.

The remaining operation in the Lynch-Welch update is division by $2$.
The same template---designing a Boolean circuit implementing this for our code, mapping non-codewords to codewords, and deriving an mc variant of the circuit---could be applied for this task.
However, in all likelihood the analog correction to the clock can be simply scaled down by construction, whereas it cannot be taken for granted that addition is straightforward in the analog domain.
In summary, all non-trivial
computations in the control loop can be
performed in a fully metastability-containing digital fashion, without
falling back to highly inefficient unary adders.

\subsection{Outline}

In \cref{sec:prelim} we define the basics of metastability-containing circuits.
Key definitions regarding precision-preserving codes are given in \cref{sec:codes}.
We show formally in \cref{sec:lowerbound} that uncertainty requires redundancy.
That is, given an (upper) bound on the uncertainty, we prove a lower
  bound on the number of redundant bits.
We provide a simple and natural $(n+k)$-bit code that uses $k$ redundant bits,
  called the ``hybrid'' code $\gHyb{n,k}$ 
  (the superscript $h$ stands for ``hybrid''). 
The $\gHyb{n,k}$ code is $k$-preserving and $\ceil{k/2}$-recoverable.
It is of broad interest to study this code as it arises naturally in the
  context of mc circuits~\cite{fuegger17aware}.
Preliminary definitions of common code schemes are given in \cref{sec:enc_examplecodes}.
The $\gHyb{n,k}$ code and the according proofs are presented in \cref{sec:hybcode}.
Informally, the first part of the code is a binary reflected Gray code of $n$ bits,
  and the second part is a unary thermometer code of $k$ bits.
The obtained code has rate $1-\BO\left(\frac{1}{n}\right)$
  for any constant $k$.
The same code $\gHyb{n,k}$ can be made $k$-recoverable by using $2k-1$ redundant
  bits.
We establish a notion of $k$-recoverable addition and show that it can be performed
  on the $\gHyb{n,k}$ code in \cref{sec:addition}.
Next, in \cref{sec:related} we discuss the connection of our code to \cite{fuegger17aware} and to error-correcting codes.

\section{Preliminaries}\label{sec:prelim}

\subsection{Ternary Logic and Metastability Containing Circuits}\label{sec:meta}

In this work, we make use of a classic extension of Boolean logic due to
  Kleene~\cite[\S64]{kleene52meta}, which allows for the presence of unspecified
  signals.
In the following, we refer to the Boolean values $\IB\coloneqq\{0,1\}$ as
\emph{stable}, while the additional third logical value $\mfu$ is the
\emph{metastable} value.
The resulting ternary set of logic values is denoted by
$\IT\coloneqq\{0,1,\mfu\}$.
Intuitively, we regard $\mfu$ as a ``superposition'' of $0$ and $1$.
Using Kleene logic is a classic way of modeling faults due to unstable
  events.

In Kleene logic, the basic gates\footnote{The specific choice of basic gates
does not matter, see~\cite{ikenmeyer18complexity}; hence, we stick to
$\AND$, $\OR$, and $\NOT$.} output a stable value if and only if the stable
inputs already determine this output. The natural extension of the basic gates
$\AND$, $\OR$, and $\NOT$ is given in \cref{tab:gates}.

\begin{table}
  \caption{Behavior of basic gates $\AND$, $\OR$, and $\NOT$.}
  \label{tab:gates}
  \centering
  \begin{tabular}{c|ccc}
    \AND & 0 & 1 & $\mfu$\\ \hline
    0 & 0 & 0 & 0\\
    1 & 0 & 1 & $\mfu$\\
    $\mfu$ & 0 & $\mfu$ & $\mfu$
  \end{tabular}
  \qquad
  \begin{tabular}{c|ccc}
    \OR & 0 & 1 & $\mfu$\\ \hline
    0 & 0 & 1 & $\mfu$\\
    1 & 1 & 1 & 1\\
    $\mfu$ & $\mfu$ & 1 & $\mfu$
  \end{tabular}
  \qquad
  \begin{tabular}{c|c}
    \NOT &\\ \hline
    0 & 1\\
    1 & 0\\
    $\mfu$ & $\mfu$
  \end{tabular}
\end{table}

A circuit $C$ is a directed acyclic graph, where each node corresponds to a basic logic gate.
By induction over the circuit structure, the behavior of basic gates defines for
any circuit the function $C\colon \IT^n\to \IT^m$ it implements.
We aim to design circuits that behave similarly to basic gates when receiving
uncertain inputs. That is, given all inputs, if changing an input bit from $0$ to $1$ does not affect the output, then setting this input bit to $\mfu$
should also not affect the output.
A circuit that satisfies this concept is called a \emph{metastability-containing (mc)} circuit.

The importance of this fundamental concept is illustrated by its surfacing in
areas as diverse as digital circuit
design~\cite{caldwell58circuits,friedrichs18containing,huffman57design},
logic~\cite{kleene52meta,korner66experience}, and
cybersecurity~\cite{hu12complexity,tiwari09flow}. Accordingly, the design of
mc circuits has received significant attention over the years. Early
on, Huffman established that all Boolean functions admit an mc
circuit~\cite{huffman57design}. However, recently unconditional exponential
lower bounds for explicit functions have been
shown~\cite{ikenmeyer18complexity}, including 
circuits of polynomial size that admit metastable upsets. Moreover, unless (the circuit variant of)
P$\,=\,$NP holds, the verifier functions of NP-complete languages have no
polynomially-sized mc circuits. These lower bounds are complemented by
a general construction by Ikenmeyer et al.~\cite{ikenmeyer18complexity} yielding circuits of size $n^{\BO(k)}|C|$, where $k$ is an upper 
bound on the number of metastable bits in its input,  and $C$ is an arbitrary circuit implementing the desired function.

\subsection{Superposition and Resolution}
To reason about the above notion of \emph{precision} (i.e., the range of values an encoded word may represent; cf.\ Section~\ref{sec:contrib}), we use the operations \emph{resolution} and \emph{superposition} as defined in~\cite{friedrichs18containing}. A precise notion of the imprecision of extended codewords
is formalized in Definition~\ref{def:code_ext}.

The superposition results in the uncertain value $\mfu$ whenever its inputs do not agree on a stable input value.

\begin{definition}[Superposition]\label{def:superposition}
Denote the \emph{superposition} of two bits by the operator $*\colon\IT \times \IT\rightarrow\IT$.
For $x,y\in \IT$, $x*y=x$ if $x=y$ and $x*y=\mfu$ otherwise. We extend the
$*$-operation to $x,y\in \IT^n$ by applying it to each index $i\in\{1,\ldots,n\}$, such that
\[ (x*y)_i = \begin{cases} x_i &\text{if } x_i=y_i,\\\mfu &\text{otherwise.} \end{cases}\]
\end{definition}

The $*$-operation is
associative and commutative, hence for $X\subseteq \IT^n$ we define $\bigstarr
X \coloneqq x_1*\hdots *x_n$, where $x_1,\ldots,x_n$ is an arbitrary
enumeration of the elements of $X$.
The second operation, the resolution, maps from ternary words to
sets of Boolean words by replacing each $\mfu$ with both stable values.

\begin{definition}[Resolution]\label{def:resolution}
Denote the \emph{resolution} by $\res\colon\IT\rightarrow\mathcal{P}(\IB)$.
For $x\in\IT$, $\res(x)=\{0,1\}$ if $x=\mfu$ and $\res(x)=\{x\}$ otherwise.
We extend this to words of length $n$ in the natural way, by setting 
\begin{equation*}
\res(x)\coloneqq \{y\in
\IB^n\,|\,\forall i\in \{1,\ldots,n\}\colon x_i\neq \mfu \Rightarrow
y_i=x_i\}\,,
\end{equation*}
for $x\in \IT^n$.
\end{definition}

In terms of this intuition, superposition may decrease precision, as the resolution of a superposition can add new words. For example, the
superposition of the words $100$ and $111$ is $1\mfu\mfu$. Resolution of this ternary word yields
$\res(1\mfu\mfu)=\{100,101,110,111\}$, so the set of possible values has grown from $\{100,111\}$ to four elements.

\subsection{Precise Codes}

Before we formally define the required properties on codes in \cref{sec:codes}, we illustrate
  them by describing the desired behavior.
Suppose that we are given two inputs of precision $k$ in $n$-bit unary encoding, which
  means that input $x\in \IT^n$ satisfies that $\res(x)$ contains at most $k+1$
  codewords, all of which are consecutive. 
  In other words, the possible values represented by $x$ form an interval of size at most $k+1$. 
Such inputs are of the form $1^*\mfu^*0^*$, where the number of $\mfu$'s is at
most $k$ and the total length of the word is $n$.
Denoting by $\gUnary{}:[n+1]\rightarrow \IT^n$ the encoding function
  for unary thermometer codes. 
  Throughout, for $M \in \IN$ we write $[M] \coloneqq \{0,\ldots,M-1\}$; 
  in particular, $[n+1] = \{0,\ldots,n\}$. 
We get that inputs are of the form
\begin{equation*}
  \bigstarr_{\ell=i}^{j}\gUnary{}(\ell)=\bigstarr_{\ell=i}^{j}1^{\ell}0^{n-\ell}
  =1^i\mfu^{j-i}0^{n-j}\,,
\end{equation*}
  where $i,j\in [n+1]$ and $0\leq j-i\leq k$.

The sum of two $n$-bit unary codes can be represented by a $(2n)$-bit unary code.
Feeding inputs $\bigstarr_{\ell=i}^{j}\gUnary{}(\ell)=1^i\mfu^{j-i}0^{n-j}$ and
  $\bigstarr_{\ell=i'}^{j'}\gUnary{}(\ell)=1^{i'}\mfu^{j'-i'}0^{n-j'}$, where
  $i,i',j,j'\in [n+1]$, into an adder that maintains precision, outputs
  $\bigstarr_{\ell=i+i'}^{j+j'}\gUnary{}(\ell)=1^{i+i'}\mfu^{j+j'-i-i'}0^{2n-j-j'}$.
This addition maintains information to the best possible degree:
We denote the set $\{i,\ldots,j\}$ by $\intvl{i,j}$.
The codewords that are resolutions of these words represent the sets
  $\intvl{i,j}$, $\intvl{i',j'}$, and $\intvl{i+i',j+j'}$, 
  respectively.~\footnote{Not all resolutions of $1^*\mfu^k 0^*$ are 
  Unary codewords. Only $k+1$ out of the $2^k$ resolutions are codewords.}
In other words, the uncertainty in the output regarding which value is
  represented matches the uncertainty of the inputs. The reader may think of this
  as a faithful representation of interval addition.
  
The above unary thermometer encoding has another crucial property.
We can consistently map any resolution of the output back to a codeword.
We define extension $\gtUnary:\IB^n\rightarrow [n+1]$ to the decoding function 
  $(\gUnary{})^{-1}:\gUnary{}([n+1])\rightarrow [n+1]$
  as follows:\footnote{
  A function $f\colon X\rightarrow Y$ is applied to a set $X'\subseteq X$ by
  application of $f$ to each element of $X'$; $f(X')\coloneqq\{f(x)|x\in X'\}$.}
For $x\in \IB^n$, let $i\in [n+1]$ be maximal such that $x_j=1$ for all $j\in [i+1]$.
Then $\gtUnary(x)\coloneqq i$ and then $
  \gUnary{}(\gtUnary(x))=1^i0^{n-i}$.
Observe that this guarantees that
$\gtUnary(\res(\bigstarr_{\ell=i}^{j}\gamma(\ell)))=\{i,\ldots,j\}$, i.e.,
  after stabilization, 
  (when each metastable bit has resolved to $0$ or $1$) %
  we can recover a value from the range that was represented
  by the word $1^i\mfu^{j-i}0^{n-j}$ using $\gtUnary$.
This property is very similar to the ability of error correction codes to recover the correct
  codeword in face of bit flips; here, we recover a ``correct'' value in the sense
  that it is consistent with the available information.
We refer to this as the code being \emph{recoverable}.
We will see that this is both highly useful and the best we can expect.

Naturally, the above positive example is of limited interest, as it is extremely
  inefficient due to the exponential overhead of unary encoding.
So, what happens if we use a non-redundant encoding? It is not hard to see that standard binary
  encoding results in a disaster.
For instance, denoting by $\gBin{n}\colon [2^n]\to\IB^n$ an $n$-bit binary 
  code,\footnote{The binary code is defined in \cref{sec:enc_examplecodes}.}
  we have that $\gBin{4}(7)=0111$ and $\gBin{4}(0)*\gBin{4}(1)=000\mfu$;
  for this pair of words, it is unclear whether their sum should be $7$ or $8$.
In contrast, $\gBin{4}(7+0)*\gBin{4}(7+1)=0111*1000=\mfu\mfu\mfu\mfu$, i.e., the resulting
  word holds no information whatsoever.
This is an extreme example, of which
  (even given all other bits) a single input bit affects all output bits.
In Kleene logic, the code provides no non-trivial guarantees in the face of uncertainty
  at its input.

The above examples show that we should aim for codes that behave similarly to
  unary thermometer encoding in the face of unstable inputs, yet have a high rate.
However, conflicting with this goal, uncertainty in the inputs requires redundancy
  of the code.
For example, a code that has no redundant bits can not recover from more
  than one $\mfu$ (see \cref{col:optrate}).
More generally speaking, a code needs to add a redundant bit for each uncertain bit
  in the input, i.e., we need $k$ redundant bits for $k$-recoverability.
We prove this claim in \cref{lem:recoverability}.

%% file: definitions.tex
\section{Codes in Ternary Logic}\label{sec:codes}

In this section, we formally define codes and related terms in ternary logic.
The second part of this section formally specifies the concepts of $k$-preserving
  and $k$-recoverable codes.
All definitions are accompanied by illustrations.

\subsection{Definition of Codes}

\begin{figure}
  \centering
  \includegraphics[width=.7\linewidth]{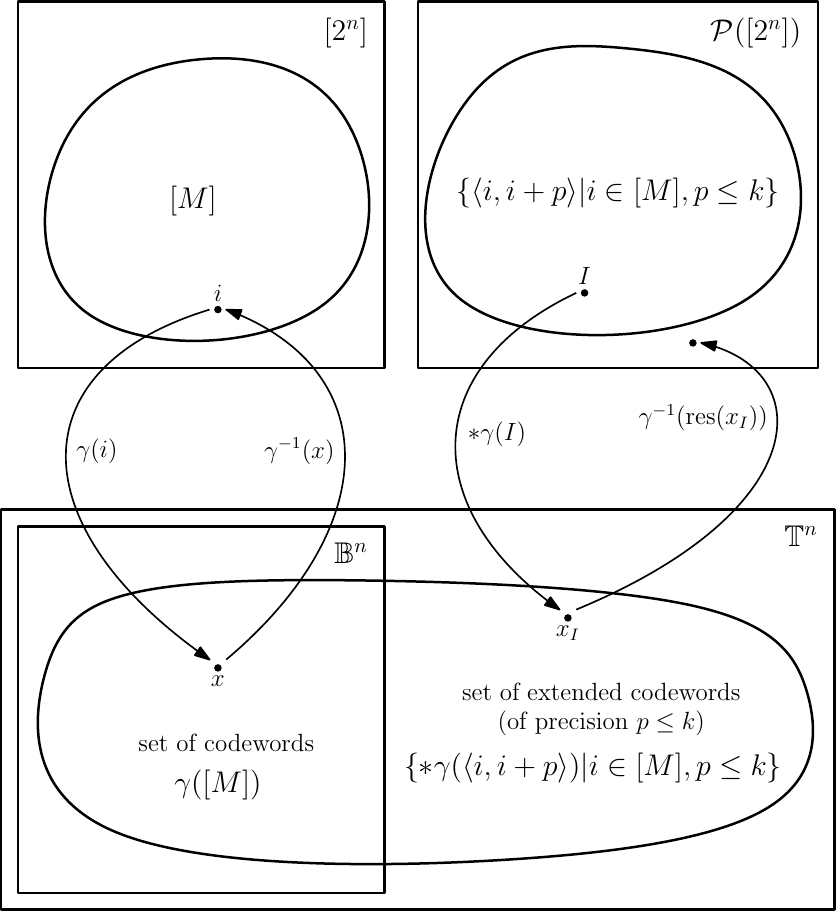}
  \caption{Overview of encoding and decoding functions and their extension to the
  ternary world.}
  \label{fig:roadmap}
\end{figure}

A code is a mapping between a finite set of numbers and binary words of a fixed length.
Any code is defined by its encoding function. The encoding function defines the set of codewords
and the decoding function.

\begin{definition}[Encoding Function]\label{def:code}
  A \emph{code} is given by an injective function $\gamma\colon[M]\rightarrow\IB^n$ that maps
  integers from a set $[M]$ to $n$-bit words. We call $\gamma$ the \emph{encoding
  function} and its image $\gamma([M])$ the \emph{set of codewords}.
  The inverse of $\gamma$ on the set of codewords, denoted by
  $\gamma^{-1}\colon\IB^n\rightarrow[M]$, is called the \emph{decoding function}.
  Every encoding function defines a decoding function.
  The decoding function is only defined on the domain $\gamma([M])$.
\end{definition}

A subset of the domain can be mapped to a ternary word
  by taking the bitwise superposition over all codewords.
We call this the extended codeword of the subset.
In this work, we always consider subsets of consecutive integers, which we call \emph{intervals}.
For integers $i \leq j$, we write $\langle i, j \rangle$ for the
interval $\{i,\ldots,j\}$. Its size is the usual cardinality $|\langle i,j \rangle| = j - i + 1$. For example, for $I = \langle i, i+p \rangle$ we have  $|I| = p+1$.
An overview of sets and mappings defined by a code is part of~\cref{fig:roadmap}.

\begin{definition}\label{def:code_ext}
  For a code $\gamma\colon[M]\rightarrow\IB^n$ and an interval $I=\intvl{i,i+p}\subseteq[M]$ we define
  \begin{center}
    \begin{tabular}{ll}
      the \emph{extended codeword} $x_I$ & by $x_I\coloneqq\bigstarr\gamma(I)$,\\
      the \emph{range} $r_{x_I}$ of $x_I$ & by $r_{x_I}\coloneqq I$, and\\
      the \emph{imprecision} of $x_I$ & by $p=|I|-1$.
    \end{tabular}
  \end{center}
\end{definition}

Depending on the encoding function, such an extended codeword might not preserve
  the information about the original set. 
It cannot be mapped back in every case.
For example, the standard binary encoding is given by 
  $\gBin{n}\colon[2^n]\rightarrow\IB^n$.
Consider the interval $\intvl{3,4}=\{3,4\}$, where codewords are given by 
  $\gBin{n}(3)=0011$ and $\gBin{n}(4)=0100$.
The extended codeword is the superposition of all codewords, it is $0\mfu\mfu\mfu$.
Resolution of the extended codeword yields more codewords,
  such that $\res(0\mfu\mfu\mfu)\neq\{\gBin{n}(3),\gBin{n}(4)\}$.

The superposition is not reversible in general because it is not an injective 
  function.
For codes, this means that there may be multiple intervals that map to the same 
  extended codeword.
Thus, mapping an interval to ternary words loses information about the original 
  interval.

The binary encoding $\gBin{n}$ maps intervals $\intvl{5,6}$ and $\intvl{4,7}$
  to the same ternary word. 
From the binary encoding, we obtain the following codewords:
  $\gBin{n}(4)=0100$,$\gBin{n}(5)=0101$,$\gBin{n}(6)=0110$, and$\gBin{n}(7)=0111$.
Hence, extended codewords of intervals $\intvl{5,6}$ and $\intvl{4,7}$ are
  given by $\intvl{5,6}=\bigstarr\{0101,0110\}=01\mfu\mfu$ and 
  $\intvl{4,7}=\bigstarr\{0100,0101,0110,0111\}=01\mfu\mfu$.

We remark that the range (as given in \cref{def:code_ext}) is not well defined.
For an extended codeword, there may be multiple ranges that map to the codeword,
  as shown in the above example. 
We use the range notation to denote the interval only when the interval is clear 
  from context.

\subsection{Preserving and Recoverable Codes}

We now define two important properties of codes that prevent loss of information when mapping intervals
to extended codewords. Parameter $k$ describes the maximum size of an interval
that can be mapped to an extended codeword without loss of information.

First, preservation ensures that the original interval is preserved, i.e.,
a resolution of an extended codeword will not add new codewords outside of the original interval.
The resolution may, however, add words that are non-codewords.
A positive and a negative illustration are given in \cref{fig:preserving}.

\begin{definition}[$k$-preserving Codes]\label{def:preserve}
  A code $\gamma$ on domain $[M]$ is $k$-preserving iff for each interval $I=\intvl{i,i+p}$,
  where $p\leq k$, the resolution $\res(x_I)$ does not contain codewords which
  are not in $\gamma(I)$.
  Formally,
  \begin{equation*}
    \gamma([M])\cap\res(x_I)=\gamma(I)\,.
  \end{equation*}
\end{definition}

\begin{figure}
  \centering
  \includegraphics[width=.7\linewidth]{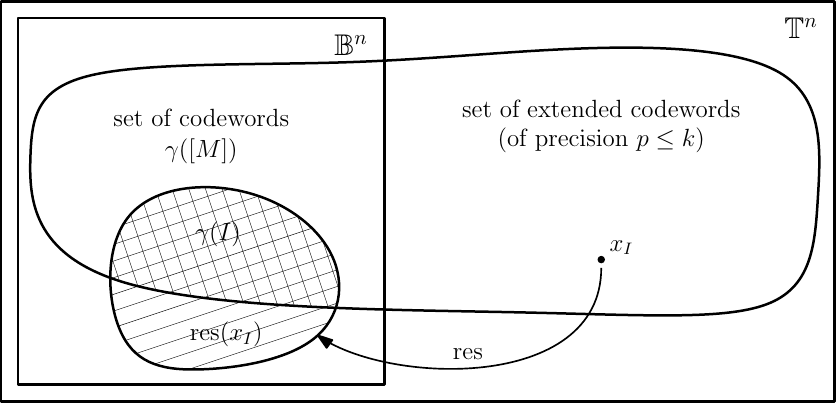}
  \includegraphics[width=.7\linewidth]{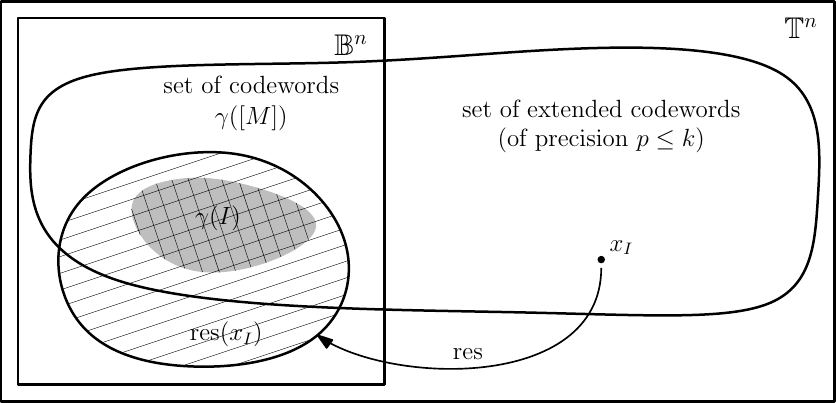}
  \caption{Illustration of a code that is preserving (top) and a code that is not
  preserving (bottom). Hatched areas denote $\res(x_I)$ and checkered areas denote $\gamma(I)$.
  For a code that is preserving, the resolution of extended codeword $x_I$ must not
  add new codewords. It may add binary words that are not codewords.}
  \label{fig:preserving}
\end{figure}

Second, the recoverability property makes sure that every non-codeword in the resolution
of an extended codeword can be mapped back to the original interval.
The property is illustrated in \cref{fig:recoverable} and defined as follows.

\begin{definition}[$k$-recoverable Codes]\label{def:recover}
  Let $x$ be an extended codeword of imprecision $p_x\leq k$.
  A code $\gamma$ is $k$-recoverable iff there exists an extension of
  $\gamma^{-1}$ from codewords to all binary words,
  such that each resolution of $x$ is mapped to the range of $x$.
  Formally, denote the extension of $\gamma^{-1}\colon\gamma([M])\rightarrow[M]$ by
  $\widetilde{\gamma}^{-1}\colon\IB^n\rightarrow[M]$.
  Then $\gamma$ is $k$-recoverable, if for $I=\intvl{i,i+p}$ and $p\leq k$,
  \begin{equation*}
    \widetilde{\gamma}^{-1}(\res(x_I))\subseteq I\,.
  \end{equation*}
\end{definition}

\begin{figure}
  \centering
  \includegraphics[width=.7\linewidth]{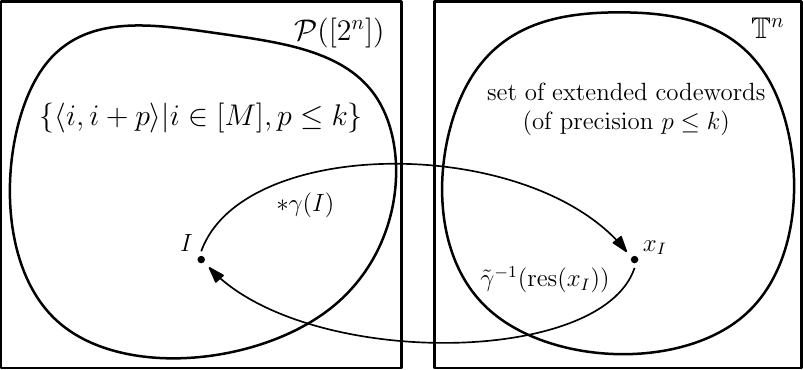}
  \caption{Illustration of the recoverability property.}
  \label{fig:recoverable}
\end{figure}
 
Preservation is a necessary condition for recoverability, that is, every
  recoverable code is also preserving.
This is shown in the following proof.

\begin{lemma}\label{lem:preserve_necessary}
  Every code that is $k$-recoverable is also $k$-preserving.
\end{lemma}
\begin{IEEEproof}
  We prove the contrapositive, i.e., no code is $k$-recoverable if it is not
  $k$-preserving.
  Assume that a code $\gamma$ is not $k$-preserving.
  By \cref{def:preserve} there is an interval $I$ such that there is some
  $x\in \gamma([M])\cap\res(x_I)$ that is not an element of $\gamma(I)$.
  By Definition~\ref{def:code} the decoding function satisfies $\gamma^{-1}(\gamma(I))=I$.
  Hence, since $x\notin \gamma(I)$ it follows that $\gamma^{-1}(x)\notin I$.
  Thus, no extension to $\gamma^{-1}$ can satisfy the requirement
  $\widetilde{\gamma}^{-1}(\res(x_I))\subseteq I$.
\end{IEEEproof}

\subsection{Efficiency of Codes}

The efficiency of a code is often measured by its \emph{rate}.
The inverse of the rate is a measure of how many redundant bits are
  added to the code.
In this work, we use the \emph{redundancy} as a measure for the 
  efficiency of a code.
The redundancy of a code is the ratio of the number of bits in use 
  compared to the minimum number necessary to encode the same 
  information~\cite{hamming:error}.
Formally, the redundancy and rate are as follows.

\begin{definition}[Redundancy and rate]\label{def:redun}
Let $\gamma \colon [2^m] \rightarrow \IB^\ell$ be a code. 
The \emph{redundancy} of $\gamma$ is given by $\ell/m \geq 1$, and the
\emph{rate} is defined as $m/\ell \leq 1$.
\end{definition}
 
A code that has redundancy $1$ is a bijection. 
Each binary word is a codeword that can be mapped to a natural number. 
From \cref{def:preserve} we can infer that no code with redundancy $1$ 
  can be more than $1$-preserving.
Consider an arbitrary interval of three elements, the corresponding 
   extended codeword must have at least $2$ unstable bits. The 
   resolution of the extended codeword thus has four elements, each 
   of them is a codeword as there are no redundant non-codewords.
As preservation is necessary for recoverability, no such code can be more than $1$-recoverable.

\begin{corollary}\label{col:optrate}
  Any code with (optimal) redundancy~$1$ is at most $1$-preserving and $1$-recoverable.
\end{corollary}

A class of codes that is of special interest to us are Gray codes.
A code is a Gray code if two consecutive codewords have Hamming 
  distance~$1$.
The Hamming distance between two words is the number of positions at which the
  bits are different.

\begin{definition}[Gray codes]\label{def:grayprop}
  A code $\gamma\colon[M]\rightarrow\IB^n$ is a Gray code if
  for $i\in[M]$ the codewords $\gamma(i)$ and $\gamma(i+1 \bmod M)$ have
  Hamming distance $1$.
\end{definition}

Given a Gray code, every extended codeword of an interval,
  of size two, has exactly one $\mfu$ bit.
The resolution of the extended codeword contains only the
  two numbers from the original interval.
We can conclude that every Gray code is at least $1$-preserving 
  and $1$-recoverable.

%% file: lower_bound.tex
\section{Recoverability Requires Redundancy}\label{sec:lowerbound}

One of our key contributions is to prove that higher recoverability requires
  higher redundancy.
We present a lower bound on the redundancy required for
  $k$-recoverability\footnote{A lower bound on the redundancy is an upper bound on the rate of a code.}
  in this section.

We show that an $n$-bit code that is $k$-recoverable can encode at most $2^{n-k}(k+1)$ values.
The redundancy of a $k$-recoverable code is bounded below by $n/(n-k+\log(k+1))$.
We first show that an extended codeword of imprecision $p\leq k$  has at least $p$ many $\mfu$'s. Recall that $M$ denotes the size of the value domain, i.e., we encode values in $[M] = \{0,\dots,M-1\}$. 

\begin{lemma}\label{lem:bitdist}
  For $k,n\in\IN$, let $\gamma$ be an $n$-bit, $k$-preserving code.
  Let $p\leq k$ and $i\in[M-p]$, then $\bigstarr
    \{\gamma(i),\ldots,\gamma(i+p)\}$ has at least $p$ many 
    $\mfu$'s.~\footnote{One would be tempted to read this as Hamming 
    distance $p$. Consider the code $100,010,001$. It is $3$-preserving, 
    all codewords have Hamming distance 2 from each other, but $\bigstarr\{100, 010,001\}=\mfu \mfu \mfu$.}
\end{lemma}

\begin{IEEEproof}
   We prove the claim by induction on $p$. The base case $p=0$ is
     trivial as codewords $\gamma(i)$ and $\gamma(i+p)$ are identical.
     They have Hamming distance $0$ and no $\mfu$.

   Next, we show the induction step from $p-1$ to $p$, where $0<p\leq k$.
   Intuitively, when going from $\gamma(i+p-1)$ to $\gamma(i+p)$, a 
     $k$-preserving code has to flip a bit that remains unchanged when 
     going from $\gamma(i)$ to $\gamma(i+p-1)$.

   Let $x_I$ be an extended codeword, with $I=\intvl{i,i+p-1}$. 
   Code $\gamma$ is $k$-preserving, thus the codeword $\gamma(i+p)$ 
     cannot be an element of $\res(x_I)$.
   There is at least one stable bit that differs for $x_I$ and $\gamma(i+p)$.
   Hence, the superposition $\bigstarr\{x_I,\gamma(i+p)\}$ has at least one 
     $\mfu$ more than $x_I$.
   This concludes the induction step.
\end{IEEEproof}

With \cref{lem:bitdist} at hand, we can show that any $k$-recoverable code can have
  at most $2^{n-k}(k+1)$ codewords.

\begin{theorem}\label{lem:recoverability}
  For $k,n\in\IN$, the domain of a code
  $\gamma\colon[M]\rightarrow\IB^n$ has at most size $M \leq 2^{n-k}(k+1)$, if $\gamma$
  is $k$-recoverable.
\end{theorem}

\begin{IEEEproof}
  Consider a $k$-recoverable code $\gamma$ on domain $[M]$ for some $M\in \IN$.
  We show that for a fixed $M = 2^{n-k}(k+1)$ the codomain must contain at least $2^n$ distinct words.
  Consider a family of intervals $I_\ell$ of size $k+1$.
  For $\ell\in[2^{n-k}]$ define $I_\ell=\intvl{\ell(k+1),\ell(k+1)+k}$.
  The family of intervals $I_\ell$ is a partition of $[M]$.
  Since $\gamma$ is $k$-recoverable, there is an extension $\widetilde{\gamma}^{-1}$
  such that $\widetilde{\gamma}^{-1}(\res(x_{I_\ell}))\subseteq I_\ell$, for all $\ell$.
  As $I_\ell \cap I_{\ell'} =\emptyset$, for $\ell,\ell'\in[2^{n-k}]$ and $\ell \neq\ell'$,
  we obtain that $\res(x_{I_\ell})\cap\res(x_{I_{\ell'}})=\emptyset$.
  Hence, there are $2^{n-k}$ many distinct intervals that have extended
  codewords with non-overlapping resolutions.

  We show that for each $I_\ell$ the resolution $\res(x_{I_\ell})$ has size at least $2^k$.
  The resolution $\res(x_{I_\ell})$ contains codewords $\gamma(i)$ and $\gamma(i+k)$.
  By~\cref{lem:bitdist}, $x_I$ has at least $k$-many $\mfu$'s. It follows that $|\res(x_I)|\geq 2^k$.
  Hence, there are $[2^{n-k}]$ many distinct intervals $I_\ell$, with resolutions $\res(x_{I_\ell})$ of size at least $2^k$.
  Accordingly, there have to be at least $2^{n-k}\cdot 2^k=2^n$ possible words in the codomain.
  Thus, if $M = 2^{n-k}(k+1)$ the codomain has at least size $2^n$.

  Finally, fix codomain $\IB^n$, i.e., a codomain of size $2^n$.
  Assume for contradiction a domain $M > 2^{n-k}(k+1)$.
  We partition $M$ into intervals of size $k+1$. There are more than $2^{n-k}$ intervals
  with one interval possibly smaller than $k+1$. By the conclusion above we use
  $2^n$ words in the codomain for the first $2^{n-k}$ intervals.
  The remaining intervals require further words, but all words in the codomain
  are allocated already.
\end{IEEEproof}

\Cref{lem:recoverability} shows that, for $k$-recoverability, an $n$-bit code $\gamma$
can encode at most $2^{n-k}(k+1)$ codewords.
We conclude that any code with domain $[2^{n-k}(k+1)]$ and codomain $\IB^n$
  is at most $k$-recoverable.

\begin{corollary}\label{col:recoverability}\label{thm:recoverability}
  For $k,n\in\IN$, no code $\gamma\colon[M]\rightarrow\IB^n$ with $|M|>2^{n-k}(k+1)$ is
  $k$-recoverable.
\end{corollary}

A direct result of \cref{thm:recoverability} is that no code with $M = 2^n$ can
be more than $1$-recoverable.
In the following section, we present two optimal rate codes, one that is
\emph{not} $1$-recoverable (standard binary code) and one that is $1$-recoverable (binary reflected Gray code).

%% file: codes.tex
\section{Common Code Schemes}\label{sec:enc_examplecodes}

We dedicate this section to introducing the following (well-known) codes: 
  standard binary code, 
  unary thermometer code, and
  binary reflected Gray code (BRGC).
They differ in their redundancy, preservation, and recoverability 
  properties.
Each property is discussed for the respective code.
The standard binary code serves as an example of a 
  well-known code that is unsuitable for $k$-recoverable codes.
It will be used in the addition; see \cref{sec:addition}.
The unary thermometer code and the BRGC are used in \cref{sec:hybcode}
  to construct the hybrid code.

  In terms of notation, we use superscripts to distinguish specific code families while keeping
$\gamma$ as a generic code symbol: $\gBin{n}$, $\gBRGC{n}$, and $\gUnary{n}$ denote the binary, binary reflected Gray code, and unary code families, respectively, and $\gHyb{n,k}$ denotes our hybrid code.

\subsection{Standard Binary Code}

We denote the encoding function of the standard binary code by $\gBin{n}$.
The standard binary code is also known as the base-$2$ numeral system.

\begin{definition}
  For $n\in\IN_{>0}$ and $i\in[2^n]$ we define the binary 
    encoding function $\gBin{n}\colon[2^n]\rightarrow\IB^n$ by
  \begin{equation*}
    \gBin{n}(i)\coloneqq \gBin{n-1}(\floor{i / 2}) \circ x_n\,,
  \end{equation*}
  where $x_n \coloneqq (i \bmod 2)$ and ``$\circ$''  denotes concatenation of bit strings. 
\end{definition}

The standard binary code has (optimal) redundancy~$1$, each number in $[2^n]$
  can be mapped to a binary word in $\IB^n$ and vice versa.
From examples in \cref{sec:intro} and \cref{sec:codes}, we observe that the code is neither $1$-preserving nor $1$-recoverable.

\subsection{Unary Thermometer Codes}

Next, we define the unary thermometer codes.
The code comes in two different flavors that we name $\gUnary{n}$
  and $\gbUnary{n}$.
The code $\gUnary{n}$ indicates the encoded value by a number of consecutive $1$'s.
Vice versa $\gbUnary{n}$ indicates the value by a number of consecutive $0$'s.
The code $\gbUnary{n}$ is obtained from $\gUnary{n}$ by flipping all
  bits. 
Both codes $\gUnary{n}$ and $\gbUnary{n}$ have redundancy $n/\log(n+1)$.

\begin{definition}[Unary Thermometer Codes]\label{def:unary}
  For $n\in\IN$ and $i\in[n+1]$ we define the encoding functions 
 \begin{align*}
   \gUnary{n}(i)&\coloneqq 1^i 0^{n-i}\,\text{, and}\\
   \gbUnary{n}(i)&\coloneqq \overline{\gUnary{n}(i)} = 0^i 1^{n-i}\,.
 \end{align*}
\end{definition}

Two example words of these codes are given by
  $\gUnary{4}(3) = 1110$, $\gbUnary{4}(3) = 0001$.

We observe that every extended codeword of imprecision $p$ has exactly
  $p$ many $\mfu$ bits.
The observation shows that in this case, each extended codeword has $i$-many
  stable bits in the front part and $(n-i-p)$-many stable bits in the rear part.

\begin{observation}\label{obs:unarystarr}
  Let $n,p\in\IN$ and $i\in[n+1]$. For any $I=\intvl{i,i+p}$, where $i+p\leq n$,
  it holds that
  \begin{align*}
    \bigstarr\gUnary{n}(I)&=\gUnary{n}(i)\starr\gUnary{n}(i+p)=1^i \mfu^p 0^{n-i-p}\,\text{, and}\\
    \bigstarr\gbUnary{n}(I)&=\gbUnary{n}(i)\starr\gbUnary{n}(i+p)=0^i \mfu^p 1^{n-i-p}\,.
  \end{align*}
\end{observation}

In particular, each resolution has the same front and rear part as the extended 
  codeword.
It follows that the resolution of an interval does not add new codewords.
Hence, both flavors of the unary thermometer codes are $n$-preserving.

\begin{corollary}\label{col:unarypreserv}
  Codes $\gUnary{n}$ and $\gbUnary{n}$ are $n$-preserving.
\end{corollary}

Next, we show that both flavors of the unary thermometer
  code are $n$-recoverable.
To this end, we first define a mapping function 
  $\mUnary{n}\colon\IB\times\IB^n\rightarrow\IB^n$.
Given an indicator bit $\pi$, the mapping maps a binary word $x$
  to a unary thermometer codeword.
If $\pi=0$, we map to $\gUnary{k}$, i.e., a codeword 
  of the form $1^\ell0^{k-\ell}$.
It is left to decide whether to choose $\ell$ according 
  to the first $0$ or the last $1$ in $x$.
Analogously, if $\pi=1$ we map to $\gbUnary{k}$ and need to
  decide whether to choose $\ell$ according to the first $1$ or
  the last $0$ in $x$.

\begin{figure}
  \centering
  \includegraphics[width=.7\linewidth]{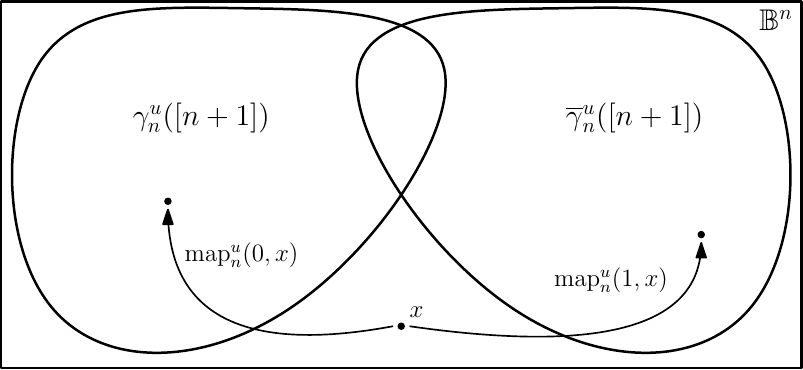}
  \caption{Illustration of the mapping $\tilde{u}$ from words to unary encoding.}\label{fig:mapping}
\end{figure}

For a binary word $x\in\IB^n$, we define indices
  $\ell_b^{\max}, \ell_b^{\min}\in [n+2]=\{0,\ldots,n+1\}$ 
  as follows.
Index $\ell_b^{\max}$ is the maximal index such that $x_{\ell_b^{\max}}=b$
  and $\ell_b^{\max}=0$ if there is no such index.
Index $\ell_b^{\min}$ is the minimal index such that $x_{\ell_b^{\min}}=b$
  and $\ell_b^{\min}=n+1$ if there is no such index.
For example, given the word $00110011$, we obtain $\ell^{min}_0=1$, $\ell^{min}_1=3$, $\ell^{max}_0=6$ and $\ell^{max}_1=8$.

\begin{definition}[Unary Thermometer Mapping]\label{def:map}
  Let $\pi\in\IB$ and $x\in\IB^n$, indices $\lmax{b}$, $\lmin{b}$ 
    are indices of $x$.
  The mapping $\mUnary{n}\colon\IB\times\IB^n\rightarrow\IB^n$ is defined by
  \begin{align*}
    &\mUnary{n}(\pi,x)\coloneqq\\
    &
    \begin{cases}
      1^{\lmin{0}-1} \circ 0^{k-\lmin{0}+1} &\text{, if }\pi=0\text{ and } x_{\ceil{n/2}}=0\\
      1^{\lmax{1}} \circ 0^{k-\lmax{1}} &\text{, if }\pi=0\text{ and } x_{\ceil{n/2}}=1\\
      0^{\lmax{0}} \circ 1^{k-\lmax{0}} &\text{, if }\pi=1\text{ and } x_{\ceil{n/2}}=0\\
      0^{\lmin{1}-1} \circ 1^{k-\lmin{1}+1} &\text{, if }\pi=1\text{ and } x_{\ceil{n/2}}=1\,.
    \end{cases}
  \end{align*}
\end{definition}

\begin{remark}
  Choosing index $\ell$ depending on $x_{\ceil{k/2}}$ is not necessary here.
  We obtain recoverability of unary thermometer codes when choosing
    $\ell$ arbitrary in $\intvl{\lmin{0}-1,\lmax{1}}$ for $\mUnary{n}(0,x)$,
    resp. $\intvl{\lmax{0},\lmin{1}-1}$ for $\mUnary{n}(1,x)$.
  The mapping $\mUnary{n}$, however, will be used in the next chapter to map
    binary words to words in the hybrid encoding.
\end{remark}

Before we show that the unary thermometer codes are recoverable, 
  we show that mapping the resolution of an extended codeword is correct.
In particular we show that, for interval $I$ and $x\in\res(x_I)$, $\mUnary{n}(0,x)$ 
  (resp.\ $\mUnary{n}(1,x)$) maps $x$ to a codeword in $\gUnary{n}(I)$
  (resp.\ $\gbUnary{n}(I)$).

\begin{lemma}\label{lem:mapping}
  Let $n,p\in\IN$, $i\in[n+1]$ and $I=\intvl{i,i+p}$, with $i+p\leq n$.
  For $x_0\in\res(\bigstarr \gUnary{n}(I))$ 
  and $x_1\in\res(\bigstarr \gbUnary{n}(I))$ 
    it holds that
  \begin{align*}
    \mUnary{n}(0,x_0) &\in \gUnary{n}(I)\,\text{, and}\\
    \mUnary{n}(1,x_1) &\in \gbUnary{n}(I)
  \end{align*}
\end{lemma}
Before we prove the statement, we first observe that the bit-wise negation of the mapping is equal to the mapping of the negated arguments.

\begin{observation}\label{obs:neg_mapping}
  Let $\pi\in\IB$ and $x_u\in\IB^k$, then
  \begin{align*}
    \overline{\tilde{u}(\pi,x_u)}=\tilde{u}(\overline{\pi},\overline{x_u})\,.
  \end{align*}
\end{observation}
\begin{IEEEproof}
  The observation can be verified by a case distinction on four cases of the
  definition. Every case follows then from the key insight that $\ell_b^{\max}$
  (resp. $\ell_b^{\min}$) of $x_u$ equals $\ell_{\overline{b}}^{\max}$
  (resp. $\ell_{\overline{b}}^{\min}$) of ${\overline{x_u}}$.
\end{IEEEproof}

We can now state the proof of \cref{lem:mapping}.

\begin{IEEEproof}
  We only show the case $\mUnary{n}(0,x_0) \in \gUnary{n}(I)$
    as the other case follows by \cref{obs:neg_mapping}.
  Assume $x_0\in\res(\bigstarr \gUnary{n}(I))$, by \cref{obs:unarystarr}
    we know that $\lmin{0}-1,\lmax{1}\in\intvl{i,i+p}$.
  Hence, $\mUnary{n}(0,x_0)=1^\ell0^{n-\ell}$ for some
    $\ell\in\intvl{i,i+p}=I$.
  The claim follows the definition of the unary thermometer code.
\end{IEEEproof}

If we define the extension to the decoding function by the mapping
  $\mUnary{n}$ we obtain recoverability of the unary thermometer codes
  immediately from \cref{lem:mapping}.

\begin{lemma}\label{lem:unaryrecov}
  Codes $\gUnary{n}$ and $\gbUnary{n}$ are $n$-recoverable.
\end{lemma}

\begin{IEEEproof}
  We prove $n$-recoverability of $\gUnary{n}$.
  We omit the proof for $\gbUnary{n}$, as it follows the same arguments.

  Define the extension of the decoding function 
    $(\widetilde{\gamma}^u_n)^{-1}\colon\IB^n\rightarrow[n+1]$
    for $x\in\IB^n$ by
  \begin{equation*}
    (\widetilde{\gamma}^u_n)^{-1}(x)\coloneqq(\gamma^u_n)^{-1}(\mUnary{n}(0,x))\,,
  \end{equation*}
  The claim follows from the first case of \cref{lem:mapping}.
\end{IEEEproof}

\begin{remark}
  Every bit is flipped when switching from $\gUnary{n}(n)$ to $\gUnary{n}(0)$,
    respectively from $\gbUnary{n}(n)$ to $\gbUnary{n}(0)$.
  This becomes an issue when using the unary thermometer code in the construction
   of our hybrid encoding.
  We can overcome this issue when appending the codewords of $\gbUnary{n}$ to the
    codewords of $\gUnary{n}$.
  This combined encoding is also a \emph{snake-in-the-box} code.
\end{remark}

\subsection{Binary Reflected Gray Code (BRGC)}

As the name suggests, the code is built from the standard binary code with a
  reflection process. 
Intuitively, for an $n$ bit code, we count through all codewords of length $n-1$ 
  with the first bit fixed to $0$. 
On the next up-count, the first bit is flipped to $1$ while fixing the remaining
  $n-1$ bits. 
  We use the term \emph{up-count} to refer to a single increment step in the counting sequence, that is, moving from the codeword for $i$ to the codeword for $i+1$. 
Afterwards, we reflect the counting order.
With the first bit fixed to $1$, we count backward through all codewords of 
  length $n-1$.  
For example, this is evident in the first two columns of \cref{tab:brgc}.
The code counts through bits three and four before flipping bit two and 
  reflecting the order.
Formally, we define~$\gBRGC{n}$, the BRGC of length $n$, by recursion.

\begin{definition}[BRGC]\label{def:BRGC}
  For $n\in\IN$ we recursively define the encoding function $\gBRGC{n}\colon[2^n]\rightarrow\IB^n$
  as follows.
  A $1$-bit BRGC is given by $\gBRGC{1}(0)\coloneqq 0$ and $\gBRGC{1}(1)\coloneqq 1$.
  For $n>1$ and $i\in[2^n]$, the encoding function $\gBRGC{n}$ is defined by
  \begin{align*}
    \gBRGC{n}(i)\coloneqq
    \begin{cases}
      0 \circ \gBRGC{n-1}(i) & \text{, if } i\in[2^{n-1}]\\
      1 \circ \gBRGC{n-1}(2^n - i - 1) & \text{, if } i\in[2^n]\setminus[2^{n-1}]\,.
    \end{cases}
  \end{align*}
\end{definition}

\begin{table}
  \centering
  \caption{Binary Reflected Gray Code $\gBRGC{4}$.}
  \begin{tabular}{cccccccc}
    $i$ & $\gBRGC{4}(i)$ & $i$ & $\gBRGC{4}(i)$ & $i$ & $\gBRGC{4}(i)$ & $i$ & $\gBRGC{4}(i)$ \\ \hline
    $0$ & $0000$ & $4$ & $0110$ & $8$ & $1100$ & $12$ & $1010$ \\
    $1$ & $0001$ & $5$ & $0111$ & $9$ & $1101$ & $13$ & $1011$ \\
    $2$ & $0011$ & $6$ & $0101$ & $10$ & $1111$ & $14$ & $1001$ \\
    $3$ & $0010$ & $7$ & $0100$ & $11$ & $1110$ & $15$ & $1000$
  \end{tabular}
  \label{tab:brgc}
\end{table}

For example, the code $\gBRGC{4}$ is listed in \cref{tab:brgc}.
The BRGC has (optimal) redundancy~$1$; each number in $[2^n]$
  can be mapped to a binary word in $\IB^n$ and vice versa.
We observe that the code $\gBRGC{n}$ indeed is a Gray code 
  according to \cref{def:grayprop}.
Hence, $\gBRGC{n}$ is $1$-preserving.

\begin{observation}\label{obs:gray_code}
  For $n\in\IN$, given two numbers $i\in[2^n]$ and $j=i+1\bmod2^n$,
  the codewords $\gBRGC{n}(i)$ and $\gBRGC{n}(j)$ have Hamming distance $1$.
\end{observation}
\begin{IEEEproof}
  We show the claim by an inductive argument over $n$, the length of the code.
  The base case $n=1$ is simple, all possible codewords (i.e.\ $0$ and $1$)
  have Hamming distance $1$.
  For the induction step ($n>1$) we make a case distinction on $i$, where
  in the case $i\in[2^n]\setminus\{2^{n-1}-1,2^n-1\}$ the claim follows from the
  recursive definition of $\gBRGC{n}$.
  In case $i=2^{n-1}-1$, we obtain that
  \begin{align*}
    \gBRGC{n}(i)&=0 \gBRGC{n-1}(2^{n-1}-1)\text{ and}\\
    \gBRGC{n}(j)&=1 \gBRGC{n-1}(2^{n-1}-1)\,.
  \end{align*}
  Hence, the codewords differ in the first bit, and the claim holds.
  In case $i=2^n-1$, we obtain that
  \begin{align*}
    \gBRGC{n}(i)&=1 \gBRGC{n-1}(0)\text{ and}\\
    \gBRGC{n}(j)&=0 \gBRGC{n-1}(0)\,.
  \end{align*}
  Hence, the codewords differ in the first bit, and the claim holds.
\end{IEEEproof}

As BRGC is a Gray code it follows that BRGC is $1$-preserving.

\begin{observation}\label{obs:brgc_preserve}
  The Binary Reflected Gray Code $\gBRGC{n}$ is exactly $1$-preserving.
\end{observation}
\begin{IEEEproof}
  Consider any interval $I=\intvl{i,i+1}$ with $i\in[M]$.
  By \cref{obs:gray_code} we know that $\gBRGC{n}(i)$ and $\gBRGC{n}(i+1)$
  differ in exactly one bit. Hence, $x_I$ has exactly one $\mfu$ bit and
  thus $\res(x_I)=\{\gBRGC{n}(i),\gBRGC{n}(i+1)\}$.

  For $\gBRGC{n}$ we have that $|M|=2^n>2^{n-2}(2+1)$.
  Hence, by \cref{col:recoverability} we follow that $\gBRGC{n}$ cannot be $2$-preserving.
\end{IEEEproof}

%% file: old_code.tex
\section{A \texorpdfstring{$k$}{k}-Preserving \texorpdfstring{$\ceil{k/2}$}{ceil(k/2)}-Recoverable Code}\label{sec:hybcode}

In this section, we define a code that is $k$-preserving and
  $\ceil{k/2}$-recoverable.
We are using a hybrid encoding scheme that uses the BRGC code
  and both flavors of the unary thermometer code.

Before defining the hybrid codes we need to state the notion of parity.
We use parity in the definition of the hybrid code to choose between
the two unary encodings.
The parity of a word denotes whether the word has an even or an odd number
of $1$'s.

\begin{definition}[Parity]\label{def:parity}
  The parity of an $n$-bit word is denoted by $\parity\colon\IB^n\rightarrow\IB$.
  It is defined by the sum over all bits modulo $2$,
  \begin{align*}
    \parity(x)\coloneqq\sum^n_{i=1}x_i \bmod 2\,.
  \end{align*}
\end{definition}

\subsection{The Hybrid Code}

\begin{table*}
  \caption{(Part of the) Hybrid Code $\gHyb{4,4}$.}\label{tab:hybrid}
  \centering
  \begin{tabular}{cccccccccccc}
    $i$ & $\gHyb{4,4}(i)$ & $i$ & $\gHyb{4,4}(i)$ & $i$ & $\gHyb{4,4}(i)$ & $i$ & $\gHyb{4,4}(i)$ & $i$ & $\gHyb{4,4}(i)$ & $i$ & $\gHyb{4,4}(i)$ \\ \hline
    & \multirow{2}{*}{\vdots} & $15$ & $0010$ $1111$ & $20$ & \cellcolor{Gray}$0110$ $0000$ & $25$ & \cellcolor{Gray}$0111$ $1111$ & $30$ & $0101$ $0000$ & $35$ & $0100$ $1111$ \\
    &  & $16$ & $0010$ $0111$ & $21$ & \cellcolor{Gray}$0110$ $1000$ & $26$ & \cellcolor{Gray}$0111$ $0111$ & $31$ & $0101$ $1000$ & $36$ & $0100$ $0111$\\
    $12$ & $0011$ $1100$ & $17$ & $0010$ $0011$ & $22$ & \cellcolor{Gray}$0110$ $1100$ & $27$ & \cellcolor{Gray}$0111$ $0011$ & $32$ & $0101$ $1100$ & $37$ & $0100$ $0011$\\
    $13$ & $0011$ $1110$ & $18$ & \cellcolor{Gray}$0010$ $0001$ & $23$ & $0110$ $1110$ & $28$ & \cellcolor{Gray}$0111$ $0001$ & $33$ & $0101$ $1110$ & \multirow{2}{*}{\vdots} \\
    $14$ & $0011$ $1111$ & $19$ & \cellcolor{Gray}$0010$ $0000$ & $24$ & $0110$ $1111$ & $29$ & \cellcolor{Gray}$0111$ $0000$ & $34$ & $0101$ $1111$ &
  \end{tabular}
\end{table*}

The $(n+k)$-bit hybrid code is denoted by
  $\gHyb{n,k}\colon[2^n(k+1)]\rightarrow\IB^{n+k}$.
It consists of an $n$-bit BRGC part and a $k$-bit
  unary thermometer part.
If the BRGC part has even parity, we count through the $\gUnary{k}$ code for
  $k+1$ up-counts.
We then increase the BRGC part by one up-count.
Thus, the parity of the BRGC part flips.
After the increment we count through the $\gbUnary{k}$ part for
  $k+1$ up-counts.
Then we do another up-count of the BRGC part and the procedure repeats.
Each time we increase the BRGC part we switch between the two flavours of the
  unary encodings.
Hence, which flavor of unary code is determined by the parity of the BRGC part.
For example, a part of the hybrid code $\gHyb{4,4}$ is given in \cref{tab:hybrid}.
Formally, the code is defined as follows.

\begin{definition}[Code $\gHyb{n,k}$]\label{def:hybrid}
  Let $n,k\in\IN$ and $n\geq k$, for $i\in [2^{n}(k+1)]$ define $x_g\in\IB^{n}$ and $x_u\in\IB^{k}$ by
  \begin{align*}
    x_g &\coloneqq \gBRGC{n}(\floor{i/(k+1)})\,,\\
    x_u &\coloneqq
    \begin{cases}
      \gUnary{k}(i\bmod(k+1)) &\text{, if }\parity(x_g)=0\\
      \gbUnary{k}(i\bmod(k+1)) &\text{, if }\parity(x_g)=1\,.
    \end{cases}
  \end{align*}
  The hybrid code is defined by the concatenation of $x_g$ and $x_u$
  \begin{align*}
    \gHyb{n,k}(i)\coloneqq x_g \circ x_u\,,
  \end{align*}
  we denote the BRGC and the unary parts by subscripts $g$ and $u$,
  \begin{align*}
    \gHyb{n,k}(i)_g &\coloneqq \gHyb{n,k}(i)_{1,n}=x_g\,,\\
    \gHyb{n,k}(i)_u &\coloneqq \gHyb{n,k}(i)_{n+1,n+k}=x_u\,.
  \end{align*}
\end{definition}

Every encoding function also defines a decoding function, as codes are injective.
The decoding function $(\gHyb{n,k})^{-1}\colon\IB^{n+k}\rightarrow[2^{n}(k+1)]$
  can be computed by
\begin{align}
  &(\gHyb{n,k})^{-1}(x) = \label{eqn:dec}\\
  &(\gBRGC{n})^{-1}(x_g) \cdot (k+1) +
  \begin{cases}
    (\gUnary{k})^{-1}(x_u) &\text{, if }\parity(x_g)=0\\
    (\gbUnary{k})^{-1}(x_u) &\text{, if }\parity(x_g)=1
  \end{cases}\,. \notag
\end{align}
Here and throughout, ``$+$'' and ``$\cdot$'' denote standard integer addition and multiplication, whereas ``$\circ$'' denotes concatenation of bit strings.

We observe two properties of extended codewords in the hybrid encoding.
First, consider an interval that does not involve an up-count in the BRGC part, i.e.,
  either $i$ is a multiple of $(k+1)$ or no element in $I$ is a multiple
  of $(k+1)$.
Then, the BRGC part of the extended codeword has no $\mfu$ and the
  unary part\footnote{For convenience we refer to the unary thermometer part as the unary part for the remainder of this work.}
  has a stable prefix and suffix.

\begin{observation}\label{obs:samecolumn}
  For $i\in [2^{n}(k+1)]$ and $p\leq k$, let $I=\intvl{i,i+p}$.
  If there are $\alpha,\ell\in\IN$ such that $i=\ell(k+1)+\alpha$
  and $\alpha + p\leq k$ then
  \begin{align*}
    (x_I)_g &= \gHyb{n,k}(i)_g = \gBRGC{n}(\ell) \,\text{, and}\\
    (x_I)_u &= \bigstarr\gUnary{k}(\intvl{\alpha,\alpha+p}) = b^\alpha \mfu^p (\overline{b})^{k-\alpha-p}\,,
  \end{align*}
  where $b=\overline{\parity((x_I)_g)}$.
\end{observation}
\begin{IEEEproof}
  Assume there are $\alpha,\ell\in\IN$ such that $i=\ell(k+1)+\alpha$
  and $\alpha + p\leq k$. Thus,
  \begin{align*}
    \floor{i'/(k+1)}=\ell \text{ for each } i'\in\intvl{i,i+p}\,.
  \end{align*}
  This already proves the first claim, i.e., each element in the interval has the same BRGC part.
  Moreover, we get that
  $\parity(\gHyb{n,k}(i')_g)=\parity(\gHyb{n,k}(i)_g)$, i.e., the unary part uses
  the same flavor of unary encoding.
  The second part then follows from \cref{obs:unarystarr}.
\end{IEEEproof}

For example, consider $n=4$, $k=p=4$, and $i=25=5(4+1)$ such that
  $I=\{25,\ldots, 29\}$ and $\alpha=0$.
We mark this as a gray region in \cref{tab:hybrid}.
The superposition yields the extended codeword $x_I=0111\,\mfu\mfu\mfu\mfu$.

Second, for each interval that does one up-count in the BRGC part, the extended unary part has at least one stable bit.
\begin{observation}\label{obs:twocolumn}
  Let $I=\intvl{i,i+p}$ and $p\leq k$ such that
    there is an element $j\in\intvl{i+1,i+p}$ that is a multiple of $(k+1)$.
  Then,   
  \begin{align*}
    (x_I)_g &= \gHyb{n,k}(i)_g \starr \gHyb{n,k}(j)_g\,\text{, and}\\
    (x_I)_u &= \mfu^\beta b^{k-\alpha-\beta} \mfu^\alpha\,,
  \end{align*}
  where $b=\overline{\parity((\gHyb{n,k}(i))_g)}$,
    $\alpha=|\intvl{i,j-1}|-1$, and $\beta=|\intvl{j,i+p}|-1$.
	Especially, if $p\leq\ceil{k/2}$ we obtain that $\alpha\leq\ceil{k/2}$ and $\beta\leq\ceil{k/2}$,
    such that $(x_u)_{\ceil{k/2}}=b$.
  Furthermore, by \cref{def:grayprop}, $(x_I)_g$ has exactly one $\mfu$ bit.~\footnote{For a bit $x \in \IB$, we denote its inversion, i.e., $1 - x$, by $\bar{x}$.} 
\end{observation}
\begin{IEEEproof}
  At most one element $j$ can be a multiple of $(k+1)$.
  From the definition of the hybrid code, it follows that all elements 
    smaller than $j$ share the same BRGC part, and
    all elements larger or equal to $j$ share the same BRGC part. 
  Hence, the first equation of the claim follows.

  Next, we show that the unary parts of the codewords of $j-1$ and $j$
    are equal.
  By \cref{obs:gray_code} we know that $\parity(\gHyb{n,k}(j-1)_g)=
    \overline{\parity(\gHyb{n,k}(j)_g)}$, i.e.,
    the parity flips when going from $j-1$ to $j$.
  Hence, we switch between both flavors of the unary encoding.
  We observe that the unary parts encode numbers $j-1 \bmod (k+1) = k$ 
    and $j \bmod (k+1)=0$.
  Due to the definition of the unary codes, i.e., $\gUnary{k}(k)=\gbUnary{k}(0)$ 
    and $\gbUnary{k}(k)=\gUnary{k}(0)$, it follows that
    $\gHyb{n,k}(j-1)_u=\gHyb{n,k}(j)_u$.

  Finally, by \cref{obs:unarystarr}, we know that 
    $(x_\intvl{i,j-1})_u=b^{k-\alpha}\mfu^{\alpha}$ and
    $(x_\intvl{j,i+p})_u=\mfu^{\beta}b^{k-\beta}$. 
  Hence, the superposition of both results in $\mfu^\beta b^{k-\alpha-\beta} \mfu^\alpha$. 
  Thus, the second equation follows.
\end{IEEEproof}

For example, consider $n=4$, $k=p=4$, and $i=18$ such that $I=\{18,\ldots, 22\}$, $j=20$, $\alpha=1$, and $\beta=2$.
We mark this as a gray region in \cref{tab:hybrid}.
The superposition yields the extended codeword $x=0\mfu10\,\mfu\mfu0\mfu$.

\subsection{Properties of the Hybrid Code}

In this section, we show that the hybrid code $\gHyb{n,k}$
  is $k$-preserving and $\ceil{k/2}$-recoverable.
Hence, the hybrid code preserves precision for intervals of 
  size up to $k$.
However, the code can only recover a codeword from 
  an extended codeword of range up to $\ceil{k/2}$.

\begin{lemma}\label{lem:preserve}
  For $n,k\in\IN$, the code $\gHyb{n,k}$ is $k$-preserving.
\end{lemma}
\begin{IEEEproof}
  Let $I=\intvl{i,i+p}$, for $i\in[2^{n}(k+1)]$ and $p\leq k$ with $i+p\leq2^n(k+1)$.
  We consider two cases that correspond to the conditions of \cref{obs:samecolumn}
    and \cref{obs:twocolumn}.

  First, assume that either $i$ is a multiple of $k+1$ or no element in I
    is a multiple of $k+1$.
  By \cref{obs:samecolumn} the BRGC part has no $\mfu$.
  All $\mfu$'s appear in the unary part.
  The claim follows from the fact that both $\gUnary{k}$ and $\gbUnary{k}$
    are $k$-preserving;
    see \cref{col:unarypreserv}.

  Second, assume there is an $j\in I$ such that $j=\ell(k+1)$ for some
    $\ell\in\IN$.
  By \cref{obs:twocolumn} the BRGC part is equal to the superposition
    of encodings of $i$ and $j$;
  \begin{align*}
    (x_I)_g = \gHyb{n,k}(i)_g \starr \gHyb{n,k}(j)_g\,.
  \end{align*}
  From $i$ to $j$ the BRGC part does one up-count.
  Due to the Gray code property, the resolution of $(x_I)_g$ contains only
    the respective parts of $i$ and $j$;
  \begin{align*}
    \res(\bigstarr(\gHyb{n,k}(I)_g)) = \{\gHyb{n,k}(i)_g, \gHyb{n,k}(j)_g\}\,.
  \end{align*}
  \Cref{obs:twocolumn} also states that the unary part has the form
  \begin{align*}
    (x_I)_u = \mfu^\beta b^{k-\alpha-\beta} \mfu^\alpha\,,
  \end{align*}
  where $\alpha=|\intvl{i,j-1}|-1$ and $\beta=|\intvl{j,i+p}|-1$.
  The resolution of $(x_I)_u$ contains all words that replace each
    $\mfu$ by either $0$ or $1$.
  Not every resolution is a codeword.
  If the BRGC part resolves to $\gHyb{n,k}(i)_g$, then
    the codewords in the resolution of $(x_I)_u$ are codewords of
    $\intvl{k-\alpha,k}$.
  Otherwise, if the BRGC part resolves to $\gHyb{n,k}(j)_g$,
    then the codewords in the resolution of $(x_I)_u$ are
    codewords of $\intvl{0,\beta}$.
  We can conclude that no codeword outside $\gHyb{n,k}(I)$ is added by the resolution
    of $x_I$.
\end{IEEEproof}

We show that the hybrid code is $\ceil{k/2}$-recoverable.
We prove recoverability by showing that there is an extension to the
  decoding function that maps all resolutions of an extended codeword back to
  the original range.
We define the extension by mapping the unary part according to the parity of
  the BRGC part.
The mapping of unary codes is done by the function $\mUnary{k}$ defined in
  \cref{def:map}.

\begin{lemma}\label{lem:recover}
  $\gHyb{n,k}$ is $\ceil{k/2}$-recoverable.
\end{lemma}
\begin{IEEEproof}
  Let $I=\intvl{i,i+p}$, for $i\in[2^{n}(k+1)]$ and $p\leq k$ with $i+p\leq2^n(k+1)$.
  We show that there is an extension to the decoding function that maps every
    resolution of $x_I$ back to interval $I$.

  The extension $\gtHyb\colon\IB^{n+k}\rightarrow[M]$ is defined by
  \begin{align*}
    \gtHyb(x)\coloneqq (\gHyb{n,k})^{-1}(x_g \circ
      \mUnary{k}(\parity(x_g),x_u))\,.
  \end{align*}

  We consider two cases that correspond to the conditions of
    \cref{obs:samecolumn} and \cref{obs:twocolumn}.

  First, assume there are $\alpha,\ell\in\IN$ such that $i=\ell(k+1)+\alpha$
    and $\alpha + p\leq k$.
  By \cref{obs:samecolumn} $(x_I)_g=\gBRGC{n}(\ell)$ has no $\mfu$.
  By the decoding function (c.f.~\eqref{eqn:dec}) the claim follows if the
    unary part is mapped to a number in $\langle\alpha,\alpha+p\rangle$.
  The mapping of the unary part depends on the parity of the BRGC part,
    let $\pi$ denote the parity of $(x_I)_g$.
  If we can show that the unary part is mapped to a codeword in
    $\gUnary{k}(\intvl{\alpha,\alpha+p})$, then the claim
    follows for the first case.
  Again, by \cref{obs:samecolumn},
    $(x_I)_u = \bigstarr\gUnary{k}(\intvl{\alpha,\alpha+p})$.
  Thus, the claim follows from \cref{obs:unarystarr}.

  Second, we assume that there is an $j\in\intvl{i+1,i+p}$ such that $j=\ell(k+1)$.
  By \cref{obs:twocolumn} the resolution of $(x_I)_g$ contains
    $\gHyb{n,k}(i)_g$ and $\gHyb{n,k}(i+1)_g$.
  W.l.o.g\ assume that $\gHyb{n,k}(i)_g$ has even parity and,
    by \cref{def:grayprop}, $\gHyb{n,k}(i+1)_g$ has odd parity.
  The proof for odd parity of $\gHyb{n,k}(i)_g$ follows the
    same line of arguments.
  From \cref{obs:twocolumn} and $p\leq\ceil{k/2}$ we obtain that
    $(y_u)_{\ceil{k/2}}=1$.
  By \cref{def:map} we deduce
    (i) $\mUnary{k}(0,y_u)=1^{\lmax{1}}0^{k-\lmax{1}}$,
    with $\lmax{1}\in\{k-(j-i),\ldots,k\}$ and
    (ii) $\mUnary{k}(1,y_u)=0^{\lmin{0}}1^{k-\lmin{0}}$,
    with $\lmin{0}\in\{0,\ldots,p-(j-i)\}$.

  It follows that for all resolutions of $x_I$, where $(x_I)_g$ is equal to
    $\gHyb{n,k}(i)_g$, the unary part is mapped to
    $\gUnary{k}(\ell)$, for $\ell\in\intvl{i \bmod (k+1), j-1 \bmod (k+1)}$.
  Similarly, for all resolutions of $x_I$, where $(x_I)_g$ is equal to
    $\gHyb{n,k}(j)_g$, the unary part is mapped to
    $\gbUnary{k}(\ell)$, for $\ell\in\intvl{j \bmod (k+1), i+p \bmod (k+1)}$.
  Thus, every resolution of $x_I$ results in a number in $I$ by using the extended
    decoding function.
\end{IEEEproof}

%% file: addition.tex
\section{\texorpdfstring{$k$}{k}-recoverable Addition}\label{sec:addition}

This section presents our circuit implementing addition under imprecision
  that is bounded by $k$.
Towards this goal, first, we formally define $k$-bit metastability-containment 
  and $k$-recoverable addition, i.e.,
  a faithful representation of interval addition due to uncertain inputs.

\subsection{Metastability-Containing Circuits}\label{sec:hazfree}

A convenient formalization of this concept is as follows. For a circuit $C$
with $n$ inputs and $m$ outputs, denote by $C(x)\in \IT^m$ the output it
computes on input $x\in \IT^n$. We say that $C$ \emph{implements} the Boolean
function $f\colon \IB^n\to \IB^m$, iff $C(x)=f(x)$ for all $x\in \IB^n$. The
desired behavior of the circuit is given by the \emph{metastable closure} of
$f$.

\begin{definition}[Metastable Closure~\cite{friedrichs18containing}]\label{def:hfree}
For function $f\colon \IB^n\to \IB^m$, denote by $f_{\mfu}\colon \IT^n\to \IT^m$
its \emph{metastable closure}, which is defined by $f_{\mfu}(x)\coloneqq
\bigstarr_{y\in \res(x)}f(y)$.
\end{definition}

The metastable closure is the ``most precise'' extension of $f$ computable by
combinational logic. It is easy to show that
$(f_{\mfu}(x))_i=\mfu$ entails that $C_i(x)=\mfu$ for any circuit $C$
implementing $f$: restricted to Boolean inputs $C$ and $f_{\mfu}$ are
identical; changing input bits to $\mfu$ can change the output bits to $\mfu$ only,
and if $f_{\mfu}(x)$ is $\mfu$ at position $i$ then $C(x)$ is $\mfu$ at position $i$.

\begin{definition}[k-Bit Metastability-Containment]
  Let circuit $C$ implement a Boolean function $f\colon\IB^n\to\IB^m$.
  We say $C$ is metastability-containing (mc) if it computes the metastable
    closure of $f$, i.e., for every $x\in\IT^n$ it holds $C(x)=f_{\mfu}(x)$.
  If $C$ is mc for every $x$ that has at most 
    $k\in\IN$ many $\mfu$s, then $C$ is $k$-bit mc.
\end{definition}

Note that if $C$ is not mc, is equivalent to: $C(x)$ outputs more
$\mfu$'s than necessary for some $x$.
The smallest non-trivial example of a circuit not implementing the metastable closure 
is given by a naive
implementation of a multiplexer circuit. A multiplexer has the Boolean
specification $\MUX(a,b,s)=a$ if $s=0$ and $\MUX(a,b,s)=b$ if $s=1$. It can be
implemented by the circuit corresponding to the Boolean formula
$\OR(\AND(a,\NOT(s)),\AND(b,s))$, resulting in a unwanted $\mfu$ at input
$(1,1,\mfu)$, cf.~\cite{ikenmeyer18complexity}.

\begin{remark}
  Note that we both use subscripts $u$ and $\mfu$, where $x_u$ denotes a certain
    part of the word $x$ and $f_\mfu$ denotes the metastable closure of a function $f$.
\end{remark}

\subsection{\texorpdfstring{$k$}{k}-Recoverable Addition}
The $k$-recoverable addition computes the encoding of the sum
  of two codewords, where inputs and outputs use the same encoding
  and the imprecision in both inputs does not exceed $k$.
We denote the addition regarding code $\gamma$ by $+_{\gamma}$.
A circuit implementing $k$-recoverable addition has to implement
  the metastable closure $(+_\gamma)_\mfu$.
We allow an arbitrary output $s\in\IT^n$ if there is an overflow.

\begin{definition}\label{lem:addition}
  Given a $k$-recoverable encoding $\gamma\colon[M]\rightarrow\IB^n$.
  Let $x,y\in\IT^n$ be extended codewords of $I_x=\intvl{i,j}$ and $I_y=\intvl{i',j'}$.
  Assume $x$ and $y$ have imprecision $p_x$ and $p_y$, such that $p_x + p_y \leq k$.
  Denote the addition of two codes by operator $+_\gamma\colon\IB^n\times\IB^n\rightarrow\IB^{n+1}$.
  We define the $k$-recoverable $\gamma$-addition of $x (+_\gamma)_\mfu y=s$ and overflow flag $\ovf$ as follows;
  If there is no overflow, i.e. $j+j'<M$, then 
  \begin{enumerate}
    \item $s$ is an extended codeword in $\gamma([M])$,
    \item $s$ has range $\intvl{i+i',j+j'}$,
    \item $s$ has imprecision $p_s = p_x + p_y$, and
    \item $\ovf=0$.
  \end{enumerate}
  If there is an overflow, i.e. $j+j'\geq M$, then $\ovf=1$.
\end{definition}

To achieve the $k$-recoverable addition of two hybrid codes in a circuit,
  we first discuss how to efficiently add hybrid codes in the Boolean world.
We add two hybrid codewords by translating the BRGC and the unary part to
  binary codewords, which can be added by any off-the-shelf binary adder.
After addition, we translate the binary representation of the sum
  back to a hybrid codeword.
Finally, we use the construction from~\cite{ikenmeyer18complexity} to obtain an mc circuit.
Note that using this construction requires us to handle non-codewords in
  the circuit input.

\subsection{Addition of Stable Codewords}

\begin{figure}
  \centering
  \includegraphics[width=.8\linewidth]{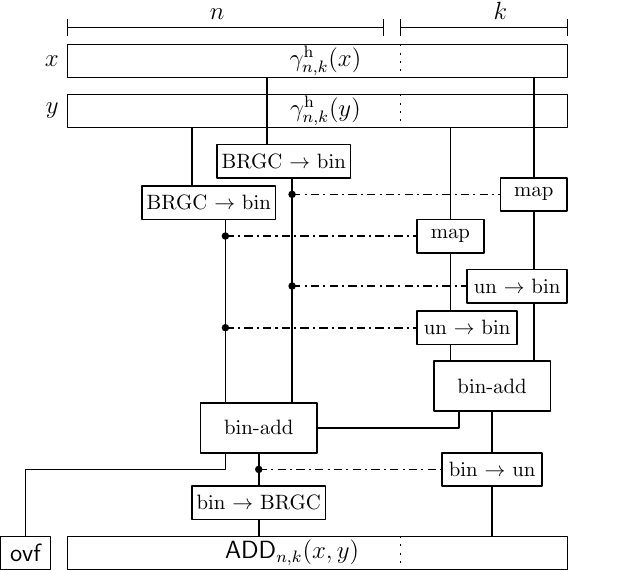}
  \caption{Addition of stable codewords and non-codewords of the hybrid code.
           The $\operatorname{map}$ circuit is only necessary for the addition of non-codewords.
           Dash-dotted lines denote that only the least significant bit is copied.}
  \label{fig:addition}
\end{figure}

We add stable codewords by separately adding the BRGC
  parts and the unary parts.
Non-codewords need to be mapped to codewords before adding them.
It is common to translate BRGC codewords to binary codewords 
  and use well-known binary adders.
We use the same technique for unary codewords.
An mc addition algorithm for hybrid codewords without translation
  to binary codewords would be desirable, because it can 
  reduce the size and depth of the circuit. 
Such an algorithm is non-trivial since there is a bidirectional 
  dependency between both parts that defies simple mc techniques.

\Cref{fig:addition} depicts our addition scheme for stable codewords and non-codewords, which we denote it by $\BLM$.
Recall that $x_g$, $y_g$ denote the BRGC part and $x_u$, $y_u$
  denote the unary thermometer part.

To translate from BRGC to binary codes and back we first observe the following
properties.

\begin{observation}\label{obs:translate}
  Let $a\in[2^n]$ and $i\in[n]$, then the following properties hold for encodings
  $\gBRGC{n}$ and $\gBin{n}$;
  \begin{align*}
    (\gBin{n}(a))_i &= \parity((\gBRGC{n}(a))_{1,i})\text{, and}\\
    (\gBRGC{n}(a))_i &= (\gBin{n}(a))_{i-1} \XOR\ (\gBin{n}(a))_{i}
  \end{align*}
\end{observation}

By using \cref{obs:translate} we formally show that there is a 
circuit that translates BRGC to binary encoding.

\begin{lemma}[BRGC $\rightarrow$ bin]\label{lem:brgctobin}
  There is a circuit of size $\BO(n)$ and depth $\BO(\log n)$ that, given input $x\in\IB^n$,
    computes the translation to binary.
  That is, it computes the output
    $o = \gBin{n}((\gBRGC{n})^{-1}(x))$.
\end{lemma}
\begin{IEEEproof}
  By \cref{obs:translate} translation to binary can be performed 
    by computing the parity for each prefix of input $x$.
  Computing an associative operation for each prefix of the input can be done 
    efficiently by a parallel prefix circuit~\cite{ladner1980parallel}.
  The parity can be computed by $\XOR$ over all bits.
  As $\XOR$ is associative, the translation can be performed by a parallel prefix $\XOR$,
    which has size $\BO(n)$ and depth $\BO(\log n)$.
\end{IEEEproof}

Conversion of the unary part to binary code depends on the
  BRGC part.
The flavor of the unary part is chosen according to the parity of the
  BRGC part.
We can spare a separate computation of the parity by observing that
  the parity is computed during the BRGC to binary translation.
The least significant bit of the binary string indicates the parity
  of the BRGC part.

\begin{observation}\label{obs:paritybit}
  Let $x\in[2^n]$, the parity of the BRGC code of $x$ is indicated by the last bit of the
    binary code of $x$;
    $\parity(\gBRGC{n}(x)) = \gBin{n}(x)[n]\,.$
\end{observation}

By using the parity we can translate the unary encoding to a binary encoding
  according to the correct flavor of unary encoding.
Binary encoding has less redundancy, a unary word of $k$ bits can be 
  represented by a binary word of $\ceil{\log(k+1)}$
  bits.

\begin{lemma}[un $\rightarrow$ bin]\label{lem:untobin}
  There is a circuit of size $\BO(k)$ and depth $\BO(\log k)$ that, for
    $\ell=\ceil{\log(k+1)}$, $x\in\gUnary{k}([k+1])$, and $\pi\in\IB$,
    computes the translation to binary.
  That is, it computes output $o$, with
  \begin{align*}
    o =
    \begin{cases}
      \gBin{\ell}((\gUnary{k})^{-1}(x)) & \text{, if } \pi=0 \\
      \gBin{\ell}((\gbUnary{k})^{-1}(x)) & \text{, if } \pi=1\,.
    \end{cases}
  \end{align*}
\end{lemma}
\begin{IEEEproof}
  By \cref{def:unary} one flavor of the unary encoding is the inverse of the respective other, 
    $\gUnary{k}(i)=\overline{\gbUnary{k}(i)}$.
  Hence, we can compute the decoding function $(\gbUnary{k})^{-1}$ by 
    using $(\gUnary{k})^{-1}$ as follows,
    $(\gbUnary{k})^{-1}(x)=(\gUnary{k})^{-1}(\overline{x})$. 
  Thus, it is equivalent to computing output
  \begin{equation*}
    o=\gBin{\ell}((\gUnary{k})^{-1}(\pi\ \XOR\ x))\,,
  \end{equation*}
  where we compute the $\XOR$ of $\pi$ and every bit of $x$.

  The translation $\gBin{\ell}((\gUnary{k})^{-1}(\cdot))$
    can be performed by a circuit that is recursively defined.
  W.l.o.g.\ we assume that $k+1$ is a power of $2$.\footnote{If $k+1$ is no power of $2$, we can append a sufficient number of bits to reach the next power of $2$, without affecting asymptotic complexity.}
  Let $m=\ceil{k/2}$ be the index of the bit in the middle of $x$.
  According to $x[m]$ the circuit recursively proceeds with all bits
    left or right of $m$. 
  Accordingly, let $m_l=\floor{k/2}$ and $m_r=\ceil{k/2}+1$ be the 
    indices left and right of the middle index.
  We define the circuit $\delta_k$ that, for $k$-bit input $x$, computes
    an $\ell$-bit output.
  For $k=1$ the circuit copies the input $\delta_1 = x$.
  For $k>1$ the circuit is defined by
  \begin{align*}
    \delta_k[1]&\coloneqq x[m]\,, \\
    \delta_k[2,\ell]&\coloneqq
    \MUX(\delta_{\ell-1}(x[1,m_l]),\delta_{\ell-1}(x[m_r,k]),x[m])\,.
  \end{align*}
  The circuit $\delta_k$ implements the unary to binary 
    translation.
  If $x[m]=1$, by \cref{def:unary} we know that $x$ encodes a value greater 
    or equal to $\ceil{k/2}$.
  Hence, the most significant bit of the binary output is set to $1$.
  If $x[m]=0$, $x$ encodes a value smaller than $\ceil{k/2}$ and
    the most significant bit is set to $0$.
  If $x[m]=1$ all bits to the left of index $m$ are $1$, and the recursive call requires bits $m_r$ to $k$.
  Similarly, if $x[m]=0$ all bits to the right of $m$ are $0$, and the recursive call requires bits $0$ to $m_l$. 
  Correctness of the unary to binary translation circuit follows by an inductive argument.

  The circuit computing the translation can be implemented by a layer
    of $k$ many $\XOR$-gates and a tree of multiplexers ($\MUX$).
  As the $\MUX$ has constant size and depth the claim follows.
\end{IEEEproof}

Using binary adders of size $n$ and $\log(k+1)$ we can add the two
  parts separately.
The carry produced by the unary parts has to be forwarded to the addition
  of the BRGC parts.
The carry produced by the BRGC parts indicates an overflow of the circuit,
  we can forward it to output $\ovf$.

After the addition, it remains to translate both parts back to their
  respective encoding.
Translation of the BRGC parts is easy because binary codes can be translated
  to BRGC with a single layer of $\XOR$ gates.

\begin{lemma}[bin $\rightarrow$ BRGC]\label{lem:bintobrgc}
  There is a circuit of size $\BO(n)$ and depth $\BO(1)$ that, given input $x\in\IB^n$,
    computes the translation to BRGC.
  That is, it computes the output
    $o = \gBRGC{n}((\gBin{n})^{-1}(x))\,.$
\end{lemma}
\begin{IEEEproof}
  By \cref{obs:translate} the translation back
    can be performed by a single layer of $\XOR$ gates.
  The circuit has size $\BO(n)$ and depth $\BO(1)$.
\end{IEEEproof}

The translation to the right unary flavor depends again on the parity
  of the BRGC part.
As stated in \cref{obs:paritybit}, the parity is already indicated by the
  least significant bit of the binary codeword.

\begin{lemma}[bin $\rightarrow$ un]\label{lem:bintoun}
  There is a circuit of size $\BO(k)$ and depth $\BO(\log k)$ that, 
    for $\ell=\ceil{\log(k+1)}$, $x\in\IB^{\ell}$, and $\pi\in\IB$
    computes the binary to unary translation,
  \begin{align*}
    o =
    \begin{cases}
      \gUnary{k}((\gBin{\ell})^{-1}(x)) & \text{, if } \pi=0 \\
      \gbUnary{k}((\gBin{\ell})^{-1}(x)) & \text{, if } \pi=1\,.
    \end{cases}
  \end{align*}
\end{lemma}
\begin{IEEEproof}
  The circuit is similar to the circuit in \cref{lem:untobin}.
  We can compute the second flavor of unary encodings by
    using $\XOR$ as follows:
  \begin{equation*}
    o=\pi\ \XOR\ \gUnary{k}((\gBin{\ell})^{-1}(x))\,.
  \end{equation*}
  W.l.o.g.\ we assume that $k+1$ is a power of $2$.\footnotemark[9]
  We define a recursive circuit $\delta_\ell$
    that computes the unary encoding.
  For $\ell=1$ the circuit is defined by $\delta_\ell(x)=x$.
  For $\ell>1$, let $\delta_{\ell-1}(x[2,\ell])=y$ be the result of the recursive subcircuit.
  The circuit adds $2^{\ell-1}$ many $0$s or $1$s to $y$ as follows:
  \begin{align*}
    \delta_\ell\coloneqq
    \MUX(y \circ 0^{2^{\ell-1}} , 1^{2^{\ell-1}} \circ y,x_1)\,,
  \end{align*}
  The circuit computes the binary to unary translation.
  If the $i$th bit in the binary word is set to $1$, then 
    the circuit adds $2^{i-1}$ many $1$s to the left of the output.
  Hence, the total number of $1$s in the output is $\sum_{i=1}^\ell 2^{i-1}\cdot x_i$.
  This matches the value encoded in binary.

  The circuit adds one $\MUX$ for each recursive call and one layer of $\XOR$ gates
    at the end.
  Hence, the circuit has asymptotically linear size and logarithmic depth in $k$.
\end{IEEEproof}

Next, we show that addition by translating each part to binary,
  adding parts separately, and translating back to the 
  respective encoding computes the addition of two hybrid codewords.
The $\BLM$ addition scheme implements the addition of hybrid codes.

\begin{lemma}\label{lem:stable_add}
  Given $x,y\in\gHyb{n,k}([M])$ the addition scheme $\BLM$ is correct;
  that is $\BLM(x,y) = x +_{\gHyb{}} y\,.$
\end{lemma}
\begin{IEEEproof}
  We prove the statement for the BRGC part and the unary part
    separately.
  Let $i = (\gHyb{n,k})^{-1}(x)$ and $j = (\gHyb{n,k})^{-1}(y)$,
    then we show that $\BLM(x,y)=\gHyb{n,k}(i+j)$.
  First, we show that $\BLM(x,y)_u = (x +_{\gHyb{}} y)_u$.
  We obtain four cases depending on the parity of $x_g$ and $y_g$.
  For each case, we have to choose the translation according to the
    flavor of the unary encoding.
  According to \cref{lem:brgctobin} and \cref{obs:paritybit}
    we obtain the correct parity bits of $x_g$ and $x_y$.
  Thus, after translation (see \cref{lem:untobin}) we obtain the
    binary representation of $i\bmod(k+1)$ and $j\bmod(k+1)$.
  Correctness of the addition follows from the correctness binary
    adder and we obtain the binary representation of $(i+j)\bmod(k+1)$.
  From \cref{obs:paritybit} and \cref{lem:bintoun} we can follow
    that the translation to unary chooses the right flavor.
  Hence, $\BLM(x,y)_u = \gHyb{n,k}(i+j)_u$, which proves the claim
    for the unary part.

  Second, we show that $\BLM(x,y)_g = (x +_{\gHyb{}} y)_g$.
  By \cref{lem:brgctobin} translation of the BRGC part to binary
    yields the binary representations of $\floor{i/(k+1)}$ and
    $\floor{j/(k+1)}$.
  From the addition of the unary parts we receive a carry that
    indicates whether $i\bmod(k+1) + j\bmod(k+1) > k+1$.
  Hence, the binary addition of the two inputs and the carry
    results in the binary representation of $\floor{(i+j)/(k+1)}$.
  Finally, by \cref{lem:bintobrgc}, we obtain that
  $\BLM(x,y)_g = \gBRGC{n}(\floor{(i+j)/(k+1)}) = \gHyb{n,k}(i+j)_n$.
\end{IEEEproof}

Given correct implementations of each of the translation circuits, \cref{lem:stable_add}
  entails that there is a circuit for the addition of two hybrid codes. 

\begin{corollary}
  There is a circuit of size $\BO(n+k)$ and depth $\BO(\log n + \log k)$ that implements
    the addition of hybrid codes.
\end{corollary}

\begin{remark}
  This scheme is far from being mc.
  Examples in \cref{sec:intro} and \cref{sec:codes} show that binary code is not even
    $1$-preserving.
  Because the $\BLM$ addition scheme heavily relies on binary codes,
    it suffers from the drawbacks of binary codes.
\end{remark}

\subsection{Addition of Extended Codewords}\label{sec:addextcodes}

We can add extended codewords by applying generic constructions for mc circuits on the above addition scheme.

When using such a construction, we need to take into account that the circuit has to
  handle inputs that are no codewords, as the resolution of an extended codeword might result in a non-codeword.
Thus, we use a circuit that implements $\gtHyb$ (from the proof of
  \cref{lem:recover}) to obtain codewords that we can add using the $\BLM$ scheme.
The extension $\gtHyb$ applies $\mUnary{k}$ to
  the unary part, given the parity of the BRGC parts.
Hence, we need a circuit that applies $\mUnary{k}$ to $x_u$ and $y_u$ before
  the binary translation.

The circuit can be built by computing all four cases of \cref{def:map}
  and selecting the correct output by a multiplexer using inputs $x_{\ceil{k/2}}$
  and $\pi$.
\begin{lemma}\label{lem:mapcircuit}
  Given inputs $x\in\IB^k$ and $\pi\in\IB$, there is a
    circuit of size $\BO(k)$ and depth $\BO(\log k)$ that
    computes $\mUnary{k}(\pi,x)$.
\end{lemma}
\begin{IEEEproof}
  The mapping $\mUnary{k}(\pi,x)$ is defined in \cref{def:map}. 
  Each of the four cases can be computed
    by a parallel prefix circuit (PPC)~\cite{ladner1980parallel}.
  However, the cases require different strategies using either 
    $\AND$ or $\OR$ as PPC operation and reading the input 
    left-to-right or right-to-left.
  We list each strategy in the following table.

  \begin{center}
    \begin{tabular}{cc|cc}
      $\pi$ & $(x_u)_{\ceil{k/2}}$ & operation & direction \\ \hline
      0 & 0 & PPC-$\AND$ & left-to-right\\ 
      0 & 1 & PPC-$\OR$ & right-to-left\\
      1 & 0 & PPC-$\AND$ & right-to-left\\
      1 & 1 & PPC-$\OR$ & left-to-right\\
    \end{tabular}
  \end{center}
  
  We can compute all four cases in parallel.
  The result is then chosen by a multiplexer with select
    bits $\pi$ and $(x_u)_{\ceil{k/2}}$.
  The full circuit consists of four PPC circuits in parallel
    and a 
    $k$-bit (4:1)-$\MUX$.
  Hence, the circuit has asymptotically linear size and 
    logarithmic depth in $k$.
\end{IEEEproof}

We can use known constructions for mc circuits to transform the circuit presented in this section into an adder that performs recoverable addition.
\begin{corollary}[of~\cite{ikenmeyer18complexity}]\label{coro:ikenmeyer}
  For $n,k\in\IN$ there is a circuit of size $(n+k)^{\BO(k)}$ and depth $\BO(k\log n + \log k)$,
    that implements $\ceil{k/2}$-recoverable addition.
\end{corollary}
\begin{proof}
    We apply the construction of  Ikenmeyer et al.\ on the combination of $\BLM$ and $\mUnary{k}$. By \cite{ikenmeyer18complexity}, this implements the metastable closure of the $\gHyb{n,k}$-addition, granted that no more than $k$ bits are unstable. Since the $\gHyb{n,k}$ code is $k$-preserving and $\ceil{k/2}$-recoverable, this is a correct implementation of $\ceil{k/2}$-recoverable addition.
\end{proof}

An example circuit execution for $n=5$ and $k=3$ is given in \cref{tab:simu_mc}.
The example shows the addition of codes with imprecision up to $2$.
We implemented the circuit in VHDL and ran a simulation using $\mfx$ to simulate
  the worst-case behavior of a metastable signal.

Similarly, we obtain a more efficient circuit by relying on masking registers.
\begin{corollary}[of~\cite{friedrichs18containing}]
  For $n,k\in\IN$ there is a circuit of size $\BO(k(n+k))$ and depth $\BO(k\log n + \log k)$ that implements $\ceil{k/2}$-recoverable addition.
\end{corollary}
\begin{IEEEproof}
Analogous to the proof of~\Cref{coro:ikenmeyer}, but using the construction from \cite{friedrichs18containing}.
\end{IEEEproof}

%% file: related.tex
\section{Further Related Work}\label{sec:related}

\subsection*{Imprecise Codes}

Gray codes are widely used as they offer $1$-recoverability,
  e.g., to recover from a small measurement error.
In retrospect, Friedrichs~et~al.~\cite{friedrichs18containing} point out that Gray
  codes offer $1$-recoverability.
However, to our knowledge, we are the first to define the $k$-recoverability
  property.
In~\cite{fuegger17aware}, it is shown that Gray code time measurements
  can be obtained by practical circuits.
Astonishingly, the measurement circuit produces $k$-recoverable
  codes identical to those presented in this work, as an intermediate result before conversion
  to $1$-recoverable Gray code.
These codes emerge naturally in this context for a completely different
  reason, suggesting that they are of general significance.

\subsection*{Metastability Containing Circuits}

Beyond the classical work on metastability-containing (a.k.a.\ hazard-free) circuits, there has
been clear recent interest in this area. Ikenmeyer et al.~\cite{ikenmeyer18complexity} study the
complexity of hazard-free circuits, showing, among other things, that hazard-free circuit complexity naturally generalizes monotone circuit complexity and that there can be large gaps between Boolean and hazard-free complexity. A central design tool in their work is a general transformation that takes a Boolean circuit and an upper bound $k$ on the number of potentially metastable input bits and produces a circuit that is mc for all such inputs. The size overhead of this construction is $n^{O(k)}$, e.g., for a constant $k$ the blow-up is polynomial in the input size. Speculation over a fixed and bounded number of unstable bits is also at the core of the metastability-tolerant design style of Tarawneh et al.~\cite{tarawnehFL17}, who integrate such speculative computing into synchronous state machines; their work can be viewed as a systems-level instance of the same basic idea that inspired
the transformation by Ikenmeyer et al.

There is also a line of recent work on the size and structure of hazard-free circuits and
formulas. Jukna~\cite{jukna21} analyzes how the requirement of being hazard-free can increase circuit size and gives explicit functions that become significantly more expensive under this constraint. Ikenmeyer, Komarath, and Saurabh~\cite{ikenmeyerKS23} develop a hazard-free Karchmer-Wigderson framework and obtain refined bounds for hazard-free formulas, showing,
for example, that even simple functions such as the multiplexer may require substantially
larger formulas in the hazard-free setting. Most recently, London, Arazi, and Shpilka~\cite{DBLP:conf/icalp/AraziS25} study hazard-free formulas for general Boolean functions, including random functions, and identify when known lower-bound techniques are tight. Together, these works
indicate that hazard-free and metastability-aware computation is an active and evolving
research direction.

%% file: conclusion.tex
\section{Conclusion}\label{sec:conclusion}
In the presence of possible bit errors, such as metastability, binary addition becomes meaningless. We proposed a new encoding scheme and a corresponding
addition procedure that maintain precision in the face of inputs that exhibit bounded total imprecision. Our hybrid code $\gHyb{n,k}$ uses $k$ redundant
bits and is $k$-preserving, but only $\lceil k/2\rceil$-recoverable: up to $k$ imprecise or metastable input bits can be tolerated without increasing the
output imprecision, yet at most $\lceil k/2\rceil$ of them can be recovered deterministically. This asymmetry between tolerated imprecision and guaranteed
recoverability reflects an inherent
gap that we leave as an interesting direction for future work.

We gave matching lower bounds showing that, for a fixed bound $k$ on the total imprecision, our code has optimal asymptotic rate. Moreover, using known black-box constructions for metastability-containing circuits, we obtain metastability-containing and purely combinational adders that implement interval addition for these codes.  For a constant $k$, this yields only polynomial overhead in circuit size compared to standard binary adders while not (asymptotically) increasing its depth.

From an application viewpoint, our framework shows that digital addition with bounded uncertainty need not fall back to unary encodings with latency linear  
in $n$. In particular, in settings where unstable inputs arise from metastability and a limited amount of output imprecision is acceptable (such as in certain control loops), our codes and adders provide a way to compute the
correct interval sum as long as the total input imprecision stays within the prescribed bound $k$. We conjecture that, when combined with synchronizer-based
designs and/or masking storage elements, such codes can be used to improve the latency–MTBF trade-off compared to purely synchronizer-based approaches.
Making this argument quantitative, however, requires optimized  circuit implementations and post-layout analysis, which we leave for future work.